%% file: main.tex
\documentclass[11pt,letterpaper]{article}

\usepackage[margin=1in]{geometry}

\usepackage{amsmath,amssymb}
\usepackage{amsthm}
\usepackage{pgf, tikz}

\usepackage{lmodern}
\usepackage[T1]{fontenc}
\usepackage{textcomp}
\usepackage[utf8]{inputenc}

\usepackage{nicefrac}

\usepackage{hyperref}
\hypersetup{
	colorlinks=true,
	citecolor = blue       % false: boxed links; true: colored links
}

\usepackage[numbers]{natbib}

\newtheorem{theorem}{Theorem}[section]

\newtheorem{lemma}{Lemma}[section]

\newtheorem{claim}{Claim}[section]
\newtheorem{claim*}{Claim}[section]

\newcommand{\RR}{\ensuremath{\mathbb{R}}}

\newcommand{\bx}{\mathbf{x}}
\newcommand{\bp}{\mathbf{p}}
\newcommand{\bv}{\mathbf{v}}
\newcommand{\bq}{\mathbf{q}}
\newcommand{\D}{\mathcal{D}}

\renewcommand{\Pr}[1]{\mbox{\rm\bf Pr}\left[#1\right]}
\newcommand{\Prr}[2]{\mbox{\rm\bf Pr}_{#1}\left[#2\right]}
\newcommand{\Ex}[2]{\mbox{\rm\bf E}_{#1}\left[#2\right]}

\newcommand{\OPT}{\mathrm{OPT}}
\newcommand{\ALG}{\mathrm{ALG}}

\newcommand{\SOLD}{\mathrm{SOLD}}

\allowdisplaybreaks

\title{An $O(\log \log m)$ Prophet Inequality for Subadditive Combinatorial Auctions}

\author{Paul D\"utting\thanks{Google Research, Brandschenkestrasse 110, CH-8002 Z\"urich, Switzerland, email: \texttt{duetting@google.com} and Department of Mathematics, London School of Economics, Houghton Street, WC2A 2AE London, UK, email: \texttt{p.d.duetting@lse.ac.uk}} \and Thomas Kesselheim\thanks{Institute of Computer Science, University of Bonn, Endenicher Allee 19a, 53115 Bonn, Germany, email: \texttt{thomas.kesselheim@uni-bonn.de}} \and Brendan Lucier\thanks{Microsoft Research, One Memorial Drive, Cambridge, MA, 02142, USA, email: \texttt{brlucier@microsoft.com}}}

\date{}

\begin{document}

\maketitle

\begin{abstract}
\input{abstract}
\end{abstract}

\input{intro}

\input{model}
\input{welfare}
\input{polytime}

\input{revenue}

\input{lower-bound}

%% REFERENCES
\bibliographystyle{plain}
\bibliography{biblio}

\appendix

\input{appendix}

\end{document}

%% file: abstract.tex
Prophet inequalities compare the expected performance of an online algorithm for a stochastic optimization problem to the expected optimal solution in hindsight.  They are a major alternative to classic worst-case competitive analysis, of particular importance in the design and analysis of simple (posted-price) incentive compatible mechanisms with provable approximation guarantees.

A central open problem in this area concerns subadditive combinatorial auctions.  Here $n$ agents with subadditive valuation functions compete for the assignment of $m$ items.  The goal is to find an allocation of the items that maximizes the total value of the assignment.  The question is whether there exists a prophet inequality for this problem that significantly beats the best known approximation factor of $O(\log m)$.

We make major progress on this question by providing an $O(\log \log m)$ prophet inequality.  Our proof goes through a novel primal-dual approach.  It is also constructive, resulting in an online policy that takes the form of static and anonymous item prices that can be computed in polynomial time given appropriate query access to the valuations.
As an application of our approach, we construct a simple and incentive compatible mechanism based on posted prices that achieves an $O(\log \log m)$ approximation to the optimal revenue for subadditive valuations under an item-independence assumption.

%% file: intro.tex
% !TEX root = main.tex
\section{Introduction}

We study the following online stochastic allocation problem.  There is a set $M$ of $m$ objects to be divided among $n$ agents.  Each agent has a valuation function that assigns a value to every subset of objects.  These valuation functions are random, drawn independently from known (but not necessarily identical) distributions.  Agents arrive one by one in an arbitrary order and when an agent arrives her valuation is revealed.  The decision-maker must choose which subset of objects to give each agent when she arrives.  The goal is to maximize the total value of the assignment.

The special case of a single object is precisely the setup of the famous prophet inequality due to Krengel, Sucheston, and Samuel-Cahn~\cite{KrengelS77,KrengelS78,SamuelC84}.  They show that there exists an online policy whose expected value is at least half of the expected optimal solution in hindsight, which in this case is simply the expected maximum value held by any agent for the object.  The policy also has a surprisingly simple form: allocate to the first agent whose value exceeds a fixed threshold calculated in advance from the known distributions.  Interest in prophet inequalities has surged recently, in part due to applications in pricing and auction design driven by the observation that the fixed threshold can be viewed as a posted price.  This has lead to a line of literature studying prophet inequalities for more general instances of the allocation problem, yielding approximately-optimal online policies and incentive compatible auctions for increasingly general problem instances with many objects and rich classes of valuation functions~(e.g.,~\cite{ChawlaHMS10,KleinbergW12,FeldmanGL15,DuttingK15,FeldmanSZ16,Rubinstein16,RubinsteinS17,DuettingFKL17,ChawlaMT19}).  
%% Maybe also EhsaniHKS18, ChawlaMT19

One of the more vexing open problems in this space concerns subadditive valuations.  A valuation function $v$ is subadditive if $v(S) + v(T) \geq v(S \cup T)$ for all sets of objects $S$ and $T$.  This captures the property that items are not complementary, in the sense that objects are not more valuable together than they are apart.  This is a natural and important property in many contexts, and the subadditive allocation problem has thus received considerable attention from both the algorithmic and economic perspectives.  For the former, there is a known $O(1)$-approximate polynomial-time algorithm for the offline problem~\cite{Feige09}.  For the latter, it is known that running a sealed-bid auction for each object separately yields an $O(1)$ approximation to the optimal welfare at any Bayes-Nash equilibrium~\cite{FeldmanFGL13}.   But despite these results, the best-known prophet inequality bound is $O(\log m)$ \cite{FeldmanGL15}.
Correspondingly, the best known incentive compatible mechanisms for Bayesian subadditive combinatorial auctions achieve $O(\log m)$ approximations to the expected welfare or revenue~\cite{FeldmanGL15,CaiZ17}.  
%\textcolor{teal}{Paul: Not sure about the following. What about submodular (and XOS) where the bound was improved from $O(\sqrt{m})$, to $O(\log(m))$, to $O(\sqrt{\log m})$, to $O((\log \log m)^3)$?  Brendan: good point.  I had intended the comment to be more specific, but the way it's currently written is too vague.  Maybe safer to back off the claim?  See proposed rewrite below.} 
%This is the most commonly-studied valuation class for which the power of truthful and non-truthful mechanisms exhibits such an asymptotic gap, and closing it is a central open challenge. 
This leaves us with a significant asymptotic gap in the known power of truthful versus non-truthful mechanisms, and closing it is a central open challenge.

We make substantial progress on this problem by obtaining an $O(\log \log m)$-approximate prophet inequality for subadditive valuations.  The policy we construct is threshold-based: the designer constructs static and anonymous prices for the objects, and each agent is assigned the set of objects (from among those that remain) that maximizes her value minus the posted prices.  The analysis is constructive, and the appropriate prices can be computed in polynomial time given access to demand queries for the valuations.\footnote{Given a valuation function $v$, a set of objects $M$, and a price $p_j \geq 0$ for all $j \in M$, a demand query returns the set $S$ that maximizes $v(S) - \sum_{j\in S}p_j$.}  This prophet inequality directly implies a polynomial-time incentive compatible posted-price mechanism for subadditive combinatorial auctions that achieves an $O(\log\log m)$ approximation to the expected optimal welfare.

Our construction can also be applied to the problem of constructing simple and approximately revenue-optimal mechanisms.  A recent line of work has studied the power of simple posted-price-based mechanisms to approximate optimal revenue in auctions with multiple items.  Recently, Cai and Zhao~\cite{CaiZ17} showed that under an item-independence assumption, it is possible to obtain an $O(\log m)$ approximation for subadditive valuations using a posted-price mechanism with either (a) an up-front entry fee for each agent, or (b) a restriction on the number of objects each agent can buy.  We show that a modification of our prophet inequality can be used to obtain an improved $O(\log\log m)$ approximation, using the same two classes of mechanisms.

\subsection{Our Results and Techniques}

\paragraph{An $\mathbf{O(\text{log}\, \text{log}\, m)}$ Prophet Inequality.} %(Section~\ref{sec:welfare})
Our first result is an existential $O(\log \log m)$-approximate price-based prophet inequality for subadditive combinatorial auctions. The pricing interpretation is that we are given access to distributions $\mathcal{D}_i$ over subadditive valuation functions $v_i$, we will precompute item prices $p_j$ for each item $j \in M$, and then agents will arrive one-by-one, each with valuation $v_i$ drawn from $\mathcal{D}_i$, and buy a subset of items $S$ that maximizes their utility $v_i(S) - \sum_{j \in S} p_j$.  Our result is that we can find prices so that the expected welfare of the resulting allocation will be an $O(\log \log m)$ approximation to the expected welfare that could be achieved by an optimal offline algorithm with advance knowledge of all valuations.

\begin{theorem}[Welfare, Existential] \label{thm:1}
For subadditive valuations drawn independently from known distributions, there exist static anonymous item prices that yield a $O(\log \log m)$ approximation to the optimal expected welfare.
\end{theorem}

We prove the theorem for complete information (known valuations), and then extend it to the Bayesian case with incomplete information. At the heart of our approach is the following lemma, which for a given and fixed subbaditive valuation $v_i$ asserts the existence of item prices $p_j$ for a given set $U$ that satisfy a certain inequality.

\begin{lemma}[Key Lemma] \label{lem:2} %[\label{lem:key-lemma}]
For every $i \in N$, subadditive function $v_i$, and set $U \subseteq M$ there exist prices $p_j$ for $j \in U$ and a probability distribution $\lambda$ over $S \subseteq U$ such that for all $T \subseteq U$
\[
	\sum_{S \subseteq U} \lambda_S \left(v_i(S \setminus T) - \sum_{j \in S} p_j\right) \geq \frac{v_i(U)}{\alpha} - \sum_{j \in T} p_j , 
\]
where $\alpha \in O(\log \log m)$.
\end{lemma}

Given this lemma it is relatively straightforward to show Theorem~\ref{thm:1}. The idea is to let $(U_1, \dots, U_n)$ be the welfare-maximizing allocation, and for each $U_i$ and $j \in U_i$ use the prices $p_j$ from Lemma~\ref{lem:2} with $U = U_i$. The welfare argument then proceeds by rewriting the welfare as the sum of buyer utilities and revenue, with Lemma~\ref{lem:2} providing a tool to lower bound the buyer utilities. 

In this lower bound argument the set $T$ from Lemma~\ref{lem:2} can be interpreted as the set of items which are already gone when we consider agent $i$, and $\lambda$ is a distribution over sets of items $S$ that agent $i$ considers to buy. Of course, agent $i$ can only buy items that are still available, so she only derives value from $S \setminus T$.  The lemma therefore establishes that the utility that can be obtained by agent $i$ is at least a factor $\alpha$ of her contribution to the optimal welfare, less the revenue obtained from selling the items from $T$.

To prove Lemma~\ref{lem:2} we write down an LP and use strong LP duality to show the following equivalent condition: There exists a probability distribution $\lambda$
 over set of items $S$ so that for every probability distribution $\mu$ with $\sum_{T: j \in T} \mu_T \leq \sum_{S: j \in S} \lambda_S$, i.e., that puts at most the same probability mass on each item $j$ as distribution $\lambda$, it holds that
\begin{align}
\sum_{S, T} \lambda_S \cdot \mu_T \cdot v_i(S\setminus T) &\geq \frac{1}{\alpha} \cdot v_i(U). \label{eq:zerosum}
\end{align}

We interpret the left-hand side as a zero-sum game, in which the protagonist chooses $\lambda$ and the antagonist chooses $\mu$, and the protagonist's goal is to maximize $\sum_{S, T} \lambda_S \cdot \mu_T \cdot v_i(S\setminus T)$.  This has a natural interpretation: the designer's goal is to find a purchasing strategy for the buyer that maximizes the value of the set they obtain, and the adversary's goal is to arrange the purchasing outcomes so that removing all previously-sold items (i.e., $T$) steals most of the value from the buyer, leaving their realized value $v_i(S \setminus T)$ as small as possible.  %Notably, the choice of prices is implicit in this formulation. They get replaced by the constraint that $\mu$ places no more probability than $\lambda$ on each item $j$.

We prove a lower bound on the value of this game by restricting attention to distributions $\lambda$ that put the same probability mass $q$ on each item.  The crux of our argument is that for each such ``equal-marginals distribution'' $\lambda$ with corresponding probability $q$, the value of the zero-sum game is at least $f(q) - f(q^2)$, where $f(q)$ is the optimal expected social welfare that can be achieved by a distribution over sets of items $S \subseteq U$ that puts probability mass at most $q$ on each item.  Intuitively, if it's possible for the adversary to choose some distribution $\mu$ over $T$ that is guaranteed to ``steal the value'' from the buyer's distribution $\lambda$ over $S$, then it must be that the set $S \cap T$ has high expected value.  But if $\lambda$ and $\mu$ each place probability at most $q$ on each item, then the distribution over $S \cap T$ places probability at most $q^2$ on each item.  Thus, if the adversary can perform well in the zero-sum game for some $q$, this directly implies that we should consider the game with the significantly smaller marginal probability $q^2$.

To turn this intuition into an $O(\log\log m)$ bound, we let the protagonist consider such ``equal marginal distributions'' for $q = 2^{-2^k}$ for $k = 0$ to $k = O(\log \log m)$, and obtain a lower bound on the value of the zero-sum game by taking the average of the sum of the corresponding lower bounds $f(q) - f(q^2)$. Now by the choice of the $q$ this sum has $O(\log \log m)$ terms, and the sum is a telescoping sum which evaluates to $f(1/2)-f(1/m^2)$. The proof is completed by observing that the latter is at least $(1/2-o(1)) \cdot v_i(U)$.

%So far this proof is existential. The main impediment to polynomial-time implementability is that the optimization over distributions has an exponential number of variables, even when restricting to a fixed marginal probability of allocating each item. Our second result addresses this limitation. 

\paragraph{Polynomial-Time Computation of Prices.} %(Section~\ref{sec:polytime})
Our second result shows how to turn this existential proof into a polytime result, assuming appropriate demand query access to the valuation functions. %The main impediment to polynomial-time implementability is that the optimization over distributions has an exponential number of variables, even when restricting to a fixed marginal probability of allocating each item.

\begin{theorem}[Welfare, Computational]
For subadditive combinatorial valuations drawn independently from known distributions and any $\epsilon > 0$, there is a polytime (in $n$, $m$, and $1/\epsilon$) algorithm to compute static and anonymous item prices for which the resulting posted-price mechanism achieves an $O(\log \log m)$ approximation to the optimal expected welfare up to an additive error of $\epsilon$.
\end{theorem}

We prove this theorem by reformulating our optimization problem in a way that avoids having to compute the equilibrium distributions in our zero-sum game, and instead draws a connection to the classic configuration LP for combinatorial assignment.  We then use the fact that the configuration LP can be solved in polynomial time using the Ellipsoid method, since a separation oracle can be implemented with demand queries. 

Specifically, we draw $q$ uniformly from the set of $q$'s introduced above. For the resulting $q$ we consider the dual LP to the configuration LP for $f(q^2)$ --- the optimal welfare that can be achieved with a probability distribution that puts probability mass at most $q^2$ on each item---and we use the prices from this dual LP scaled by $q$.

The intuition behind this construction is as follows.  In our argument above we bounded the value of the zero-sum game by $f(q) - f(q^2)$.  Here $f(q)$ is the highest expected value that the protagonist could obtain from a choice of $\lambda$ if the adversary abstained, and $f(q^2)$ is an upper bound on how much value the antagonist can take away by choosing $\mu$ optimally.
By taking the dual prices for the configuration LP for $f(q^2)$ and scaling them by $q$, we are effectively setting prices that approximate the welfare loss due to the antagonist's strategy, which is to say the worst-case loss from excluding items that have already been sold.

\paragraph{Revenue Maximization.} % (Section~\ref{sec:revenue})
We also show how to leverage our new insights
%concerning item prices for subadditive valuations, 
to make progress on another important frontier in algorithmic mechanism design. Namely, the question of how well the optimal revenue that can be obtained by a Bayesian incentive compatible (BIC) mechanism can be approximated by a simple and dominant strategy incentive compatible (DSIC) mechanism.
Our revenue approximation makes use of a framework for constructing simple mechanisms due to Cai and Zhao~\cite{CaiZ17}, which builds upon a recent literature applying a duality approach to revenue maximization~\cite{CaiDW16}. Cai and Zhao established an $O(\log m)$ approximation under a natural item independence assumption.  
Under the same assumption we show:

\begin{theorem}[Revenue]
When buyers have subadditive valuations over independent items, there is a simple DSIC mechanism that yields an $O(\log \log m)$ approximation to the optimal BIC revenue.
\end{theorem}

A key step in the proof of Cai and Zhao~\cite{CaiZ17} invokes a posted-price-based prophet inequality for welfare maximization, and indeed their approximation factor of $O(\log m)$ is driven by the $O(\log m)$-approximate prophet inequality that they apply. However, one cannot apply a prophet inequality to their framework as a black box. The reason is that the prophet inequality is invoked to argue that the value of a certain interim allocation rule---the core of a revenue-optimal mechanism---can be approximated by posted prices. 
%One might hope to improve this approximation by plugging in our new prophet inequality instead, but unfortunately one cannot apply a prophet inequality to their framework as a black box.
%
%In particular, one important requirement for the approach is the ability to impose arbitrary (not necessarily equal) constraints on the marginal probability of allocating each item, guided by a revenue-optimal interim allocation rule.  We therefore provide an alternative proof of our prophet inequality that permits such constraints, and obtains a bound relative to the optimal offline assignment that satisfies the same constraints.  
%\textcolor{teal}{Brendan: say more about our techniques here?}
%
We obtain the improved bound by extending our prophet inequality so that it can handle arbitrary (not necessarily equal) constraints on the marginal probability of allocating each item, and using this to obtain a better price-based approximation to the core. 

\paragraph{Going Beyond $\mathbf{O(\text{log}\, \text{log}\, m}$).} %(Section~\ref{sec:lower-bound})
%\item Para 10: We show that our approach is tight.  We leave with a conjectured method for improving to $O(1)$ approx.
Finally, we demonstrate that the $O(\log\log m)$-factor that shows up in all our bounds is best possible using our approach.  In particular, our analysis restricts to distributions that set the same marginal probability $q$ of allocating each item.  We show by way of example that such distributions (and their associated dual prices) can suffer loss as high as $\Omega(\log\log m)$.  

\begin{theorem}[Lower Bound]
There exists a subadditive valuation function $v_i$ over $m$ items such that for any $q \in [0,1]$, any $\alpha \in o(\log \log m)$, and any distribution $\lambda$ that puts probability at most $q$ on each item there exists a distribution $\mu$ that puts probability at most $q$ on each item that violates inequality~(\ref{eq:zerosum}).
\end{theorem}

%\textcolor{orange}{[I think it would be fun to say something about the example here ... a couple of the set-cover integrality gap functions stacked on top of each other ...]}

Our restriction to equal-marginal distributions was crucial for our %iterative 
approach to optimizing over distributions.  Of course, it is natural to wonder whether our bound could be improved by relaxing the equal-marginals assumption and permitting an arbitrary profile of marginal distributions.  Indeed, we conjecture that an $O(1)$-approximate prophet inequality can be achieved using item prices that are dual to a distribution with unequal marginals.  But we leave resolving this conjecture as an open problem.

\paragraph{Discussion: Connection to Balanced Prices.}

The main difference between our $O(\log \log m)$ approximation and the earlier state-of-the-art $O(\log m)$ prophet inequality from \cite{FeldmanGL15} is that this earlier approach constructed prices by approximating subadditive valuations through fractionally subadditive (a.k.a.~XOS) functions.  This leads to ``balanced prices'' in the sense of \cite{DuettingFKL17}, where the sum of all prices matches the optimal allocation precisely.

More generally, the balanced prices framework of \cite{DuettingFKL17} entails constructing prices for any fixed valuation profile, such that (a) the prices of any subset of items partially offset the value lost due to allocating these items, and (b) the sum of all prices is upper bounded by the total value of the optimal allocation of all items. With parameters $1/\alpha$ and $\beta$ for (a) and (b) this leads to $O(1/(\alpha\beta))$ price-based prophet inequalities. However, such balanced prices cannot lead to a better than $O(\log m)$ approximation for subadditive combinatorial auctions~\cite{FeldmanGL15}.

%\textcolor{teal}{BL: lots of edits to the following paragraph; it could use another pass.}

Our prices are different, and will generally be much higher.  The basic intuition is that, under balanced prices, the sum of all prices approximates the optimal welfare, so in a sense the set of all items is ``affordable'' and the prices facilitate an outcome where most of the items are purchased.  However, depending on the curvature of the subadditive valuations, it may be better to target much smaller sets for purchase if they already capture most of the value.  By looking at different marginals $q$ we are basically considering different sizes $q \cdot m$ of sets of items to go after, and our key lemma establishes that there is always a good choice of $q$.  As $q$ becomes smaller, the prices we construct are tailored to facilitate purchases of smaller sets of items, and hence the item prices tend to increase.

\subsection{Further related work}

%\textcolor{teal}{Some of these will be cited in the intro above}
%
%\begin{enumerate}
%\item Algorithmic results on subadditive optimization
%\item Prophet inequality literature, including ``combinatorial prophet inequalities'' and ``subadditive prophet inequalities.''  Also reference the contention resolution scheme approach.
%\item The $O(\log m / \log\log m)$ result, as concurrent work.
%\item Literature on revenue maximization in mechanism design.  Core/tail decomposition, Shuchi + student paper on posted prices, duality framework (Cai-Devanur-Weinberg)
%\end{enumerate}

From a purely algorithmic perspective social welfare with fractionally subadditive (or XOS) and subadditive valuations can be approximated to within a constant factor assuming demand queries. The state-of-the-art for both XOS and subadditive valuations is a $1-1/e$ approximation due to \cite{Feige09}. These approximation guarantees are best possible in the sense that they match the integrality gap of the LP formulations they are based on.  

An important question in algorithmic mechanism design concerns the gap between the best (worst-case) approximation guarantee that can be obtained without incentives (i.e., purely algorithmically) and with a truthful mechanism.
Recent breakthroughs for submodular (a subclass of XOS) valuations were obtained by Dobzinski \cite{Dobzinski16}, who gave a $O(\sqrt{\log m})$ truthful approximation mechanism for submodular valuations and by Assadi and Singla~\cite{AssadiS19} who gave a $O((\log \log m)^3)$ truthful approximation mechanism for XOS valuations. %Our work yields a truthful mechanism for Bayesian subadditive combinatorial auctions with a $O(\log \log m)$ approximation guarantee.
Finding DSIC approximation mechanisms with constant worst-case approximation guarantees for either XOS or subadditive valuations or disproving their existence is a major open problem.

A related question concerns the relative power of truthful direct-revelation mechanisms and general mechanisms at equilibrium.  The latter can be analyzed using the price of anarchy framework, in which the expected optimal solution is compared with the worst-case expected outcome at any Bayes-Nash equilibrium of the mechanism.  Christodoulou, Kovacs, and Schapira~\cite{christodoulou2016bayesian} established an $O(1)$ price of anarchy bound for simultaneous item auctions for XOS valuations, which was subsequently extended to a variety of auction formats and related solution concepts~\cite{SyrgkanisT13}.  In particular, Feldman, Fu, Gravin, and Lucier~\cite{FeldmanFGL13} established an $O(1)$ price of anarchy for simultaneous item auctions under subadditive valuations.

Prophet inequalities for XOS and subadditive combinatorial auctions in which agents arrive one by one were previously given in \cite{FeldmanGL15,DuettingFKL17} and \cite{EhsaniHKS18}. For XOS valuations an optimal factor 2 is shown in \cite{FeldmanGL15,DuettingFKL17}, and this can be improved to $1-1/e$ by additionally assuming agents arrive in random order \cite{EhsaniHKS18}. For subadditive valuations, Feldman et al.~\cite{FeldmanGL15} give an $O(\log m)$ approximation. 
Rubinstein and Singla \cite{RubinsteinS17} consider a related but different problem, where there is one subadditive function across all entities that arrive over time. They give an $O(\log m \log^2 r)$ prophet inequality, where $r$ is the rank of an arbitrary downward closed feasibility constraint.

In concurrent and independent work, \cite{Zhang20} was able to improve the $O(\log m)$ prophet inequality for subadditive combinatorial auctions to $O(\log m/\log \log m)$.  This marks an important breakthrough as it shows that it is possible to improve upon the $O(\log m)$ bound.

Our application to revenue maximization builds upon a recent literature on approximately revenue-optimal mechanisms for buyers with multi-dimensional types.  For unit-demand buyers, one can obtain a constant approximation to the optimal mechanism with multiple buyers~\cite{chawla2007algorithmic,ChawlaHMS10,chawla2010power}.  Simple constant approximations are known for additive buyers with independent valuations, using a technique known as a tail-core decomposition which bounds separately the revenue contribution from rare outlier values and from ``expected'' valuation profiles~\cite{HartNisan2017,LiYao13,BILW14,yao2014n}.  Chawla and Miller showed how to combine both approaches to develop a general class of approximately optimal mechanisms based on posted prices with per-buyer entry fees~\cite{chawla2016mechanism}.  Cai, Devanur, and Weinberg further unify these approaches using a flexible duality framework to effectively ``linearize'' valuations with respect to revenue~\cite{CaiDW16}.  The ideas behind these mechanisms have since been extended to more general valuation classes, including XOS and subadditive valuations~\cite{CaiZ17,rubinstein2018simple}.  Most related to the current paper is the work of Cai and Zhang~\cite{CaiZ17}, which (among other things) uses this framework to design an $O(\log m)$-approximate mechanism for subadditive valuations, based on posted item prices with per-buyer entry fees.

%% file: model.tex
% !TEX root = main.tex

\section{Model and Definitions}

\paragraph{Subadditive Combinatorial Auctions.} We are given a set $N$ of $n$ buyers and a set $M$ of $m$ goods. 
% We typically use index $i$ to refer to a buyer, and index $j$ to refer to an item.
Each buyer $i \in N$ has a valuation function $v_i: 2^M \rightarrow \RR_{\geq 0}$, which is assumed to be normalized and monotone, i.e., $v_i(\emptyset) = 0$ and $v_i(S) \leq v_i(T)$ for $S \subseteq T \subseteq M$. A valuation function $v_i$ is \emph{subadditive} if 
\[
v_i(S) + v_i(T) \geq v_i(S \cup T) \quad \text{for} \quad S, T \subseteq M.
\]

We use $\bv = (v_1, \dots, v_n)$ to denote a vector of valuation functions. We will occasionally write $\bv = (v_i, \bv_{-i})$, where we use $\bv_{-i}$ to denote the valuations of all buyers except buyer $i$.

We assume a \emph{Bayesian setting}, in which the valuation function of each buyer $i$ is drawn independently from distribution $\D_i$. We write $\D = \prod_i \D_i$ for the joint distribution. We emphasize that the independence here is \emph{across bidders}. Valuations of a fixed agent can be arbitrarily correlated across items (though we revisit this when discussing applications to revenue maximization in Section~\ref{sec:revenue}). We assume that the designer knows the distributions from which the valuation functions are drawn, but not the realizations of the random draws.

An allocation $\bx = (x_1, \dots, x_n)$ defines for each buyer $i \in N$ a set of goods $x_i \subseteq M$ that he receives. We require that no good is assigned more than once, i.e., that $x_i \cap x_j = \emptyset$ whenever $i \neq j$. We write $\bx_{<i}$ for partial allocations to buyers $s < i$, i.e., for $t \geq i$ we have $x_t = \emptyset$.

We evaluate allocations by the welfare they achieve. The welfare of an allocation $\bx$ is $\sum_i v_i(x_i)$. We write $\OPT(\bv)$ for the welfare-maximizing allocation, and $\bv(\OPT((\bv)) = \sum_{i = 1}^{n} v_i(\OPT_i(\bv))$ for the welfare it achieves.

\paragraph{Posted-Price Mechanisms.}

%For notational convenience we introduce the notation $\bx_{<i}$ for a partial allocation in which $x_\ell = \emptyset$ for $\ell \geq i$.

A posted-price mechanism uses a set of functions $p_i(\cdot \mid \bx_{<i}): 2^M \rightarrow \RR_{\geq 0}$ which assign a non-negative price to each set of items $S \subseteq M$.  Note that these functions can be personalized, they can be ``per set'' rather than ``per item'' (i.e., the price of a set of items need not be a sum of prices of individual items), and they may depend on which items were already allocated.

Of particular interest will be posted-price mechanisms that use static anonymous item prices. %\textcolor{teal}{BL: I prefer not using commas here, as in the last sentence, but I didn't change throughout.}  
A posted-price mechanism has \emph{anonymous} prices if there exists a single set of functions $p(\cdot \mid \bx_{<i})$ such that $p_i(S \mid \bx_{<i}) = p(S \mid x_{<i})$ for all $i, S$, and $\bx_{<i}$. It uses \emph{item prices} if $p_i(S \mid \bx_{<i}) = \sum_{j \in S} p_i(\{j\} \mid \bx_{<i})$ for all $i, S$, and $\bx_{<i}$. Finally, prices are \emph{static} if for each $i$ there is a single function $p_i(\cdot)$ such that $p_i(S \mid \bx_{<i}) = p_i(S)$ for all $S$ and $\bx_{<i}$.

An important advantage of posted-price mechanism with static anonymous item prices is that they can be succinctly described by a single 
vector $\bp = (p_1, \dots, p_m) \in \mathbb{R}^m_{\geq 0}$. 

A posted-price mechanism proceeds as follows.  The buyers arrive sequentially, and for notational convenience we assume they are indexed according to their arrival order.\footnote{All of our results continue to hold in a more general setting where the arrival order is arbitrary and unknown to the designer (but still fixed in advance) and is revealed online as the buyers arrive.} Upon arrival of buyer $i$ the mechanism posts a price $p_i(S, \bx_{<i})$ for each set of items $S$. Buyer $i$ buys any set of items $x_i$ that maximizes her utility $u_i(x_i,\bp) = v_i(x_i) - \sum_{j \in x_i} p_j$ among all such sets.

%We assume bidders have access to a \emph{demand oracle} that given prices and a set of items that are still available returns a utility maximizing set of items.  \textcolor{teal}{BL: I'm not sure we should go down the road of modeling the buyer-side computation. I think it's enough to just say buyers maximize utility. Or do we mean the seller here?}

Given a fixed choice of item prices, we will tend to write $\ALG$ for the corresponding posted-price mechanism, $\ALG(\bv)$ for the resulting allocation of items when valuations are $\bv$, and use $\bv(\ALG(\bv)) = \sum_{i=1}^{n} v_i(\ALG_i(\bv))$ to denote the welfare it achieves.

\paragraph{Prophet Inequalities.}

We will follow the ``prophet-inequality paradigm'' to evaluate the performance of posted-price mechanisms. That is, we will evaluate the performance of a posted-price mechanism $\ALG$ by comparing its expected welfare $\Ex{\bv \sim \D}{\bv(\ALG(\bv))}$ to the expected optimal welfare $\Ex{\bv \sim \D}{\bv(\OPT(\bv))}$. Extending the notion of competitive ratio from the worst-case analysis of online algorithms, we define the \emph{(stochastic) competitive ratio} of a posted-price mechanism as
\[
\sup_{\D} \frac{\Ex{\bv \sim \D}{\bv(\OPT(\bv))}}{\Ex{\bv \sim \D}{\bv(\ALG(\bv))}}.
\]

%% file: welfare.tex
% !TEX root = main.tex

\section{An $\mathbf{O(\text{log}\, \text{log}\, m)}$ Price-Based Prophet Inequality for Welfare}
\label{sec:welfare}

We start by establishing the existence of an $O(\log \log m)$-competitive price-based prophet inequality for subadditive combinatorial auctions and the goal of maximizing welfare. 

\begin{theorem}\label{thm:upper-bound}
For subadditive combinatorial auctions there is a $O(\log \log m)$-competitive posted-price mechanism that uses static anonymous item prices. 
\end{theorem}

It suffices to show Theorem~\ref{thm:upper-bound} for $m > 2$ as for $m = O(1)$ the competitive ratio is constant. We will prove the theorem in two steps. In Section~\ref{sec:complete-info}, we show the claim for complete information. That is, we assume valuations are fixed and known. In Section~\ref{sec:incomplete-info}, we prove it for the Bayesian case with incomplete information.  %\textcolor{teal}{Did we define complete information?}

%% Subsection: Proof for Complete Information
\subsection{Proof for Complete Information}
\label{sec:complete-info}

Our key lemma and driver of the improved competitive ratio is the following lemma. We prove this lemma using LP-duality, and derive the existence of appropriate prices and the corresponding probability distribution over sets of items through a zero-sum game formulation.  Lemma~\ref{lem:key-lemma} is a restatement of Lemma~\ref{lem:2} from the introduction.  Recalling the discussion after Lemma~\ref{lem:2}, the intuition is that the revenue raised by selling items that would typically be allocated to buyer $i$ (set $T$), plus the utility that buyer $i$ can obtain from the remaining items (by buying set $S \backslash T$), approximates buyer $i$'s contribution to the expected optimal welfare (i.e., $v_i(U)$).

\begin{lemma}\label{lem:key-lemma}
For every $i \in N$, subadditive function $v_i$, and set $U \subseteq M$ there exist prices $p_j$ for $j \in U$ and a probability distribution $\lambda$ over $S \subseteq U$ such that for all $T \subseteq U$
\[
	\sum_{j \in T} p_j + \sum_{S \subseteq U} \lambda_S \left(v_i(S \setminus T) - \sum_{j \in S} p_j\right) \geq \frac{v_i(U)}{\alpha}, 
\]
where $\alpha \in O(\log \log m)$.
\end{lemma}

Before we prove Lemma~\ref{lem:key-lemma}, let's see how it implies the desired result.

\begin{proof}[Proof of Theorem~\ref{thm:upper-bound} (complete information)]
Let $\OPT(\bv) = (U_1, \dots, U_n)$ be the welfare-maximizing allocation for valuations $\bv$. 
%\textcolor{blue}{Paul: Do we really need the following assumption?} Note that because valuations are monotone we can without loss of generality assume that the $U_i$ are not only disjoint, but also cover $M$.  
Define a vector $\bp$ of item prices as follows: For $i \in N$ and $j \in U_i$ use price $p_j$ from Lemma~\ref{lem:key-lemma}.  
Let $\ALG$ be the posted-price mechanism that uses prices $\bp$. Denote the allocation of $\ALG$ on valuation profile $\bv$ by $\ALG_1(\bv), \dots, \ALG_n(\bv)$ and let $\SOLD(\bv) = \cup_{i=1}^{n} \ALG_i(\bv)$ denote the set of items sold by the mechanism.

To derive a lower bound on the welfare achieved by the posted-price mechanism we will use that the welfare can be decomposed into utility and revenue. Namely, if we write $u_i((v_i,\bv_{-i}),\bp)$ for the utility of buyer $i$ and $r(\bv,\bp)$ for the revenue, then
\[
	\bv(\ALG(\bv)) = \sum_{i=1}^{n} u_i((v_i,\bv_{-i}),\bp) + r(\bv,\bp).
\]

We begin by deriving a lower bound on the sum of the utilities. To this end consider an arbitrary buyer $i$. 
Let $\lambda$ be the probability distribution over sets of items $S \subseteq U_i$ from Lemma~\ref{lem:key-lemma} and let $T = \cup_{\ell < i} \ALG(\bv)_\ell \cap U_i \subseteq \SOLD(\bv) \cap U_i$. 
%Let $S \subseteq U_i \setminus T$ be as in Lemma~\ref{lem:key-lemma}. 
Now because buyer $i$ could draw a set of items $S$ from $\lambda$ and buy set $S\setminus T$ (or no set at all if this gives her negative utility),
\[
	u_i((v_i,\bv_{-i}),\bp) \geq \sum_{S \subseteq U_i} \lambda_S \left(v_i(S \setminus T) - \sum_{j \in S} p_j\right) \geq \frac{v_i(U_i)}{\alpha} - \sum_{j \in T} p_j,
	%u_i((v_i,\bv_{-i}),\bp) \geq v_i(S \setminus T) - \sum_{j \in S \setminus T} p_j \geq \frac{v_i(U_i)}{\alpha} - \sum_{j \in T} p_j,
	%u_i \geq \frac{1}{2} (v_i(S\setminus T) - \sum_{j \in S \setminus T} p_j) \geq \frac{v_i(U)}{4\alpha} - \sum_{j \in T} p_j,
\]	
where the inequality holds by Lemma~\ref{lem:key-lemma}.

Summing over all buyers $i$, we obtain
\begin{align}
\sum_{i = 1}^{n} u_i((v_i,\bv_{-i}),\bp) \geq \frac{\bv(\OPT(\bv))}{\alpha} - \sum_{j \in \SOLD(\bv)} p_j. \label{eq:utility}
\end{align}

On the other hand, the revenue obtained by the posted-price mechanism is
\begin{align}
r(\bv,\bp) = \sum_{j \in \SOLD(\bv)} p_j. \label{eq:revenue}
\end{align}

Adding (\ref{eq:utility}) and (\ref{eq:revenue}) shows the claim.
\end{proof}

To prove Lemma~\ref{lem:key-lemma} we will write down an LP that captures the claim. The lemma statement will be satisfied whenever the optimal solution to the LP has non-negative value, and we will show that this is indeed the case using strong duality.

Consider an arbitrary buyer $i$ and an arbitrary set $U$. Let $\gamma = 1/\alpha$. In order to establish Lemma~\ref{lem:key-lemma}, we have to show that there are prices $p_j$ for $j \in U$ and a distribution $\lambda$ over sets of items $S \subseteq U$ such that for all $T \subseteq U$
%\[
%	\sum_{j \in T} p_j + \beta \left(\sum_{S \subseteq U} \lambda_S \Big( v_i(S\setminus T) - \sum_{j \in S} p_j \Big)\right) \geq \gamma v_i(U).
%\]
\[
	\sum_{j \in T} p_j + \sum_{S \subseteq U} \lambda_S \left( v_i(S\setminus T) - \sum_{j \in S} p_j \right) \geq \gamma v_i(U).
\]
or equivalently
\begin{align}
	\sum_{S \subseteq U} \lambda_S \sum_{j \in S} p_j - \sum_{j \in T} p_j \leq \sum_{S \subseteq U} \lambda_S v_i(S\setminus T) - \gamma v_i(U). \label{sufficient-condition}
\end{align}
%\[
%	\beta \sum_{S \subseteq U} \lambda_S \sum_{j \in S} p_j - \sum_{j \in T} p_j \leq \beta \sum_{S \subseteq U} \lambda_S v_i(S\setminus T) - \gamma v_i(U).
%\]

To show inequality \eqref{sufficient-condition}, we will consider the following LP. The LP is for a fixed $\lambda$ and has variables $p_j \geq 0$ for $j \in U$ and two more variables $\ell_+ \geq 0$ and $\ell_{-} \geq 0$. The extra variables model a slack term $\ell_{+} - \ell_{-}$ of arbitrary sign which we consider adding to the the left-hand side of the inequality. The LP maximizes the slack term.
\begin{align*}
\text{max}\quad&\ell_{+} - \ell_{-}\\
\text{s.t.}\quad&\sum_{S \subseteq U} \lambda_S \sum_{j \in S} p_j - \sum_{j \in T} p_j + (\ell_{+} - \ell_{-}) \leq \sum_{S \subseteq U} \lambda_S v_i(S\setminus T) - \gamma v_i(U) &&\text{for all $T \subseteq U$}\\ %\\ %\text{for all $T \subseteq M$}\\
& p_j \geq 0 &&\text{for all $j \in U$}\\
& \ell_{+}\geq 0\\
&\ell_{-} \geq 0.
\end{align*}

As a non-negative slack means that inequality \eqref{sufficient-condition} is fulfilled, we know that there are prices fulfilling inequality \eqref{sufficient-condition} if and only if this LP has an optimal solution with non-negative objective value.

Our strategy for showing this will be to go through the dual. Indeed by strong duality it is equivalent to show that every feasible solution to the following dual LP with variables $\mu_T\geq 0$ for $T \subseteq U$ has non-negative value. 
\begin{align*}
\text{min}\quad& \sum_T \mu_T \left(\sum_{S \subseteq U} \lambda_S v_i(S\setminus T) - \gamma v_i(U) \right)\\
\text{s.t.}\quad& - \sum_{T: j \in T} \mu_T + \sum_T \sum_{S: j \in S} \lambda_S \mu_T \geq 0 &&\text{for all $j \in U$}\\ %\text{for all $j \in M$}\\
&\sum_{T} \mu_T = 1\\
&\mu_T \geq 0 &&\text{for all $T \subseteq U$.}
%&\mu_T \geq 0 &&\text{for all $T \subseteq U$}
\end{align*}

We note that the dual constraints are equivalent to $\sum_T \mu_T  = 1$ and $\sum_{T: j \in T} \mu_T \leq \sum_{S: j \in S} \lambda_S$ for all $j \in U$. So they naturally define probability distributions over sets of items $T \subseteq U$, which can put at most the same probability mass on each item $j \in U$ as the probability distribution $\lambda$.

This means that the LP has non-negative value if and only if for every probability distribution $\mu$ with $\sum_{T: j \in T} \mu_T \leq \sum_{S: j \in S} \lambda_S$ it holds that
\begin{align}
\sum_{S, T} \lambda_S \mu_T v_i(S\setminus T) &\geq \gamma v_i(U). \label{pre-minimax}
\end{align}

We can now formulate the search for an appropriate $\lambda$ as a zero-sum game in which pure strategies correspond to subsets of items. The maximizing player chooses $S$, the minimizing player chooses $T$, and the payoff associated with two sets $S$ and $T$ is $v_i(S\setminus T)$. We want to show that there is mixed strategy $\lambda$ for the maximizing player such that when the minimizing player is constrained to use a mixed strategy $\mu$ which puts at most the same probability mass on each item $j$ as $\lambda$, then the value of the game is at least $\gamma v_i(U)$.

To this end let $\bq = (q_1, \dots, q_{|U|})$ where $q_j \in [0,1]$ and let $\Delta(\bq)$ denote all probability distributions $\nu$ over sets $S \subseteq U$ such that $\sum_{T \ni j} \nu_T \leq q_j$ for all $j$. Define
\begin{align}
g(\bq) = \max_{\lambda \in \Delta(\bq)} \min_{\mu \in \Delta(\bq)} \sum_{S,T \subseteq U} \lambda_S \mu_T v_i(S\setminus T).
\end{align}

Inequality \eqref{pre-minimax} and hence inequality \eqref{sufficient-condition} and Lemma~\ref{lem:key-lemma} are therefore equivalent to there being a $\bq$ such that $g(\bq) \geq v_i(U)/O(\log \log m)$.

The following lemma shows that it is in fact possible to achieve this with a uniform vector $\bq$ in which $q_i = q_j$ for all $i$ and $j$.

\begin{lemma}\label{lem:loglog-complete}
There exists a $q \in [0,1]$ such that for $\bq = (q,\dots,q) \in [0,1]^{|U|}$ we have 
\[
	g(\bq) \geq \frac{1}{O(\log \log m)} v_i(U).
\]
\end{lemma}
\begin{proof}
For $q \in [0,1]$ define
\begin{align*}
	g(q) &= \max_{\lambda \in \Delta(q,\dots,q)} \min_{\mu \in \Delta(q,\dots,q)} \sum_{S,T \subseteq U} \lambda_S \mu_T v_i(S\setminus T), \quad\text{and}\\
	f(q) &= \max_{\lambda \in \Delta(q,\dots,q)} \sum_{S \subseteq U} \lambda_S v_i(S).
\end{align*}

We now use subadditivity, that $\sum_T \mu_T = 1$, and finally that $\mu_T\lambda_S$ defines a probability distribution on $S \cap T$ that puts at most probability mass $q^2$ on each item to obtain
\begin{align*}
	g(q) 
	&= \max_{\lambda \in \Delta(q,\dots,q)} \min_{\mu \in \Delta(q,\dots,q)} \sum_{S,T \subseteq U} \lambda_S \mu_T v_i(S\setminus T)\\
	&\geq \max_{\lambda \in \Delta(q,\dots,q)} \min_{\mu \in \Delta(q,\dots,q)} \sum_{S,T \subseteq U} \lambda_S \mu_T \Big(v_i(S) - v_i(S \cap T)\Big)\\
	&= \max_{\lambda \in \Delta(q,\dots,q)} \Bigg( \sum_{S \subseteq U} \lambda_S v_i(S) - \max_{\mu \in \Delta(q,\dots,q)} \sum_{S,T \subseteq U} \lambda_S \mu_T v_i(S \cap T)\Bigg)\\
	&\geq \max_{\lambda \in \Delta(q,\dots,q)} \Bigg( \sum_{S \subseteq U} \lambda_S v_i(S) - \max_{\nu \in \Delta(q^2,\dots,q^2)} \sum_{S \subseteq U} \nu_S v_i(S)\Bigg)\\
	&= f(q) - f(q^2),
\end{align*}

For any $\ell$ we thus have,
\begin{align*}
\sum_{i=0}^{\ell} g(2^{-2^i}) \geq \sum_{i=0}^{\ell} \Bigg( f\Big(2^{-2^i}\Big) - f\Big(2^{-2^{i+1}}\Big)\Bigg) = f(2^{-1}) - f(2^{-2^{\ell+1}}),
\end{align*}
by a telescoping sum argument.

With $\ell = \log \log m$,
\begin{align*}
	2^{-2^{\ell+1}} = 2^{-2 \log m} = \frac{1}{m^2}.
\end{align*}

Hence for $\ell = \log \log m$,
\begin{align*}
\sum_{i=0}^{\ell} g(2^{-2^i}) \geq f\Big(\frac{1}{2}\Big) - f\Big(\frac{1}{m^2}\Big).
\end{align*}

We conclude the proof by showing a lower bound on $f(1/2)$ and an upper bound on $f(1/m^2)$. A lower bound on $f(1/2)$ follows from the fact that $\lambda$ could take the set $U_i$ with probability $1/2$. So
\begin{align*}
	f\left(\frac{1}{2}\right) &\geq \frac{1}{2} \cdot v_i(U).
\end{align*}
For the upper bound on $f(1/m^2)$ we exploit the trivial upper bound on $v_i(S)$ namely $v_i(U)$. Using this we obtain,
\begin{align*}
	f\left(\frac{1}{m^2}\right) 
	&= \max_{\lambda \in \Delta(1/m^2,\dots,1/m^2)} \sum_{S \subseteq U} \lambda_S v_i(S)\\
	&\leq \max_{\lambda \in \Delta(1/m^2,\dots,1/m^2)} \left(\sum_{S \subseteq U, S \neq \emptyset} \lambda_S \right)v_i(U)\\
	&\leq \sum_{j \in U} \left(\sum_{S: S \ni j} \lambda_S \right) v_i(U)\\
	& \leq m \cdot \frac{1}{m^2} \cdot v_i(U) \\
	&= \frac{1}{m} \cdot v_i(U).
\end{align*}

We conclude that
\begin{align*}
\max_q g(q) \geq \frac{1}{\ell+1} \sum_{i=0}^{\ell} g(2^{-2^i}) \geq \frac{1}{\ell+1} \left(\frac{1}{2} - \frac{1}{m}\right) v_i(U),
\end{align*}
as claimed.
\end{proof}

%% Subsection: Proof for Incomplete Information
\subsection{Proof for Incomplete Information}
\label{sec:incomplete-info}

We next show how to extend our arguments to the incomplete information case, where the valuations are not fixed and known but rather are drawn from known distributions. Our proof is based on the following variant of Lemma~\ref{lem:key-lemma}.
The proof of this lemma in Appendix~\ref{app:proof-key-lemma-bayesian} follows the same basic steps as the proof of Lemma~\ref{lem:key-lemma}, but requires some additional care when deriving a lower bound on the value of the zero-sum game.  In particular, since the zero-sum game now has payoffs that are defined with respect to the distribution over valuations, our argument requires that we relate the value of the game to the expected value of the distribution over full-information games.
%\textcolor{teal}{BL: Since...?  Can we add a sentence or two about the technical difficulty?}

\begin{lemma}\label{lem:key-lemma-bayesian}
For every probability distribution $\D = \prod_i \D_i$ over subadditive valuation functions, there exist prices $p_j$ for $j \in M$ and probability distributions $\lambda^{i,\bv}$ over $S \subseteq M$ for all $i$ and $\bv$ such that for all $T \subseteq M$
%such that for all $T \subseteq U$ there is $S \subseteq U\setminus T$ for which
\[
	\sum_{j \in T} p_j + \Ex{\bv}{\sum_{i=1}^{n} \sum_{S} \lambda^{i,\bv}_S \left(v_i(S \setminus T) - \sum_{j \in S} p_j\right)} \geq \frac{1}{\alpha} \Ex{\bv}{\bv(\OPT(\bv))}, 
	%\sum_{j \in T} p_j + \frac{1}{2} \left(v_i(S) - \sum_{j \in S} p_j\right) \geq \frac{v_i(U)}{4 \alpha}, 
\]
where $\alpha \in O(\log \log m)$.
\end{lemma}

Using Lemma~\ref{lem:key-lemma-bayesian} it is straightforward to prove Theorem~\ref{thm:upper-bound} using a hallucination trick similar to that in the price of anarchy literature \cite{SyrgkanisT13,Roughgarden15}, the literature on algorithmic stability \cite{HardtRS16}, and in the balanced prices framework for prophet inequalities \cite{DuettingFKL17}. We provide a formal proof in Appendix~\ref{app:proof-thm-upper-bound} for completeness.

%% file: polytime.tex
% !TEX root = main.tex

\section{Computing Prices in Polynomial Time}
\label{sec:polytime}

Theorem~\ref{thm:upper-bound} shows the existence of static anonymous item prices that yield an $O(\log \log m)$ approximation to the optimal welfare. In this section, we establish the following computational version of this result. It shows how to achieve the same approximation guarantee in polynomial time. The proof also reveals an alternative for choosing prices that yield an $O(\log \log m)$ approximation. We assume to have access to demand oracles for the valuation functions in the support of $\mathcal{D}$. Recall that a demand oracle for valuation function $v_i$ takes as input item prices $p_1, \ldots, p_m$ and returns the set $S \subseteq M$ that maximizes $v_i(S) - \sum_{j \in S} p_j$.

\begin{theorem}\label{thm:upper-bound-polytime}
For subadditive combinatorial auctions and any $\epsilon > 0$, there is a polynomial-time (in $n$, $m$, and $1/\epsilon$) algorithm to compute static and anonymous item prices for which the resulting posted-price mechanism achieves expected welfare at least $\frac{1}{\alpha} \cdot \Ex{\bv}{\bv(\OPT(\bv)} - \epsilon$ where $\alpha = O(\log \log m)$. 
\end{theorem}

Our proof is based on the following version of Lemma~\ref{lem:key-lemma-bayesian} that includes computations in polynomial time using demand oracles.

\begin{lemma}
\label{lem:key-lemma-bayesian-polytime}
For every probability distribution $\D = \prod_i \D_i$ over subadditive valuation functions and every $\epsilon > 0$, there exist prices $p_j$ for $j \in M$ and probability distributions $\lambda^{i,\bv}$ over $S \subseteq M$ for all $i$ and $\bv$ such that for all $T \subseteq M$
\[
	\sum_{j \in T} p_j + \Ex{\bv}{\sum_{i=1}^{n} \sum_{S} \lambda^{i,\bv}_S \left(v_i(S \setminus T) - \sum_{j \in S} p_j\right)} \geq \frac{1}{\alpha} \Ex{\bv}{\bv(\OPT(\bv))} - \epsilon, 
\]
where $\alpha \in O(\log \log m)$.  Moreover, assuming that $\bv(\OPT(\bv)) \leq 1$ for all $\bv$ in the support, for any $\zeta > 0$, there is an algorithm that uses $\mathrm{poly}(n,m,1/\epsilon,\log(1/\zeta))$ demand oracle queries that computes such prices with probability at least $1 - \zeta$.
\end{lemma}

\begin{proof}
We will proceed in two steps. First, we will assume to have the ability to compute expectations $\Ex{\bv}{\,\cdot\,}$. Then, in the second step, we simulate this ability by sampling from the distributions and bound the errors.

\paragraph{Computation of Prices.} Fix some valuation profile $\bv$.  As in the case of complete information, we will consider an $f$-function that maximizes the expected value $v_i(S)$, where $S$ is drawn from a constrained distribution over sets of items.  To this end, for each $q \in [0,1]$ let $\Gamma(q) = \{\{\nu^{i}\}_{i=1}^{n} \mid \sum_{i=1}^{n} \sum_{S: S \ni j} \nu^{i}_S \leq q \;\text{for all $j$ in $M$}\}$ be the collection of distribution profiles for which the marginal probability that each item is allocated is at most $q$.  Then we will define
%Define
%\[
%g^{\bv}(q) = \sup_{\lambda \in \Gamma(q)} \inf_{\mu \in \Delta(q,\dots,q)} \sum_{T} \mu_T \left(\sum_{i=1}^{n} {\sum_S \lambda^{i}_S v_i(S \setminus T)} \right).
%\]
%
%As before we will define 
\begin{align*}
	f^{\bv}(q) &= \sup_{\lambda\in \Gamma(q)} \sum_{i=1}^{n} \sum_{S} \lambda^{i}_S v_i(S).
\end{align*}

%In the following, we will first keep $q$ fixed. To this end, let $\ell = \log \log m$, fix some $X \in \{0, \ldots, \ell\}$ and set $q = 2^{-2^X}$.

%For a fixed valuation profile $\bv$, 
Fix $q \in [0,1]$ and let $\lambda^{i,\bv,q}$ be some choice of $\lambda \in \Gamma(q)$ that achieves the maximum in the definition of $f^{\bv}(q)$. 

We now consider the following linear program:
\begin{align*}
	\text{max}\quad& \sum_{i=1}^{n} \sum_{S} x^i_S v_i(S) \\
	\text{s.t.}\quad& \sum_i \sum_{S: j \in S} x^i_S \leq q^2 && \text{ for all $j \in M$} \\
	& \sum_S x^i_S \leq 1 && \text{ for all $i$} \\
	& x^i_S \geq 0 && \text{ for all $i$, $S$}
\end{align*}
The value of this program is precisely equal to $f^{\bv}(q^2)$, from the definition of $f^{\bv}$. The dual of this linear program is given by
\begin{align*}
	\text{min}\quad& \sum_{j \in M} q^2 y_j + \sum_{i=1}^{n} u_i \\
	\text{s.t.}\quad& u_i + \sum_{j \in S} y_j \geq v_i(S) && \text{ for all $i$, $S$} \\
	& u_i \geq 0 && \text{ for all $i$} \\
	& y_j \geq 0 && \text{ for all $j$}
	\end{align*}

Let $(y^{\bv,q}_j)_{j \in M}$, $(u^{\bv,q}_i)_{i \in N}$ denote an optimal solution to this dual.  We note that this optimal solution can be computed in polynomial time. Indeed, it suffices to find a separation oracle for the dual program.  A separation oracle query is equivalent to finding a set $S$ that maximizes $v_i(S) - \sum_{j \in S}y_j$ for a given choice of $(y_j)_{j \in M}$.  But this is precisely a demand query, interpreting the dual variables $y_j$ as item prices.  Thus, given access to a demand oracle for $v_i$, one can solve the dual program and compute $(y^{\bv,q}_j)_{j \in M}$.

Our algorithm now chooses $q$ so as to maximize $\Ex{\bv}{f^{\bv}(q)} - \Ex{\bv}{f^{\bv}(q^2)}$ among all  $q = 2^{-2^X}$ for some $X \in \{0, \ldots, \ell\}$, where $\ell = \log\log m$. We then define our prices according to $p^q_j = \Ex{\tilde{\bv}}{q y^{\tilde{\bv},q}_j}$.  That is, $p^q_j$ is equal to the expected value of $y_j^{\tilde{\bv},q}$ scaled by $q$.

\paragraph{Approximation Guarantee.} To prove the approximation guarantee, we first show a property of the prices that holds regardless of the choice of $q$.

\begin{claim}
\label{claim:polytime-q-qsquard}
For any choice of $q \in [0, 1]$, the prices defined by $p^q_j = \Ex{\tilde{\bv}}{q y^{\tilde{\bv},q}_j}$ fulfill
\begin{align*}
\Ex{\bv}{\sum_{i, S} \lambda^{i,\bv,q}_S \left( v_i(S \setminus T) - \sum_{j \in S} p^q_j \right) + \sum_{j \in T} p^q_j} \geq \Ex{\bv}{f^{\bv}(q)} - \Ex{\bv}{f^{\bv}(q^2)}.
\end{align*}
\end{claim}

\begin{proof}[Proof of Claim~\ref{claim:polytime-q-qsquard}]
For all $\bv$ and $T$, we have
\begin{align*}
&\sum_{i, S} \lambda^{i,\bv,q}_S \left( v_i(S \setminus T) - \sum_{j \in S} p^q_j \right) + \sum_{j \in T} p^q_j \\
&\qquad\geq \sum_{i, S} \lambda^{i,\bv,q}_S \left( v_i(S) - v_i(S \cap T) - \sum_{j \in S} p^q_j \right) + \sum_{j \in T} p^q_j \\
&\qquad= \sum_{i, S} \lambda^{i,\bv,q}_S v_i(S) - \sum_{i, S} \lambda^{i,\bv,q}_S v_i(S \cap T) - \sum_{i, S} \lambda^{i,\bv,q}_S \sum_{j \in S} p^q_j + \sum_{j \in T} p^q_j,
\end{align*}
where the inequality uses subadditivity and the equality is simply expansion.

We now analyze each of the terms of this expression.  First, $\sum_{i,S} \lambda^{i,\bv,q}_S v_i(S) = f^{\bv}(q)$ from the definition of $\lambda^{i,\bv,q}$.  Next, since $(y^{\bv,q}_j)_{j \in M}, (u^{\bv,q}_i)_{i \in N}$ satisfy the dual program, it must be that $v_i(S \cap T) \leq u^{\bv,q}_i + \sum_{j \in S \cap T} y^{\bv,q}_j$.  This implies 
\[ \sum_{i, S} \lambda^{i,\bv,q}_S v_i(S \cap T) \leq \sum_i u^{\bv,q}_i \sum_S \lambda^{i,\bv,q}_S + \sum_{i, S} \lambda^{i,\bv,q}_S \sum_{j \in S \cap T} y^{\bv,q}_j \leq \sum_i u^{\bv,q}_i + q \sum_{j \in T} y^{\bv,q}_j. \] 
Furthermore $\sum_{i, S} \lambda^{i,\bv,q}_S \sum_{j \in S} p^q_j \leq q \sum_{j \in M} p^q_j$. We therefore conclude
\begin{align*}
&\sum_{i, S} \lambda^{i,\bv,q}_S \left( v_i(S \setminus T) - \sum_{j \in S} p^q_j \right) + \sum_{j \in T} p^q_j \\
&\qquad\geq f^{\bv}(q) - q \sum_{j \in T} y^{\bv,q}_j - \sum_i u^{\bv,q}_i - q \sum_{j \in M} p^q_j + \sum_{j \in T} p^q_j.
\end{align*}
We now take expectations over $\bv$ to get
\begin{align*}
&\Ex{\bv}{\sum_{i, S} \lambda^{i,\bv,q}_S \left( v_i(S) - v_i(S \cap T) - \sum_{j \in S} p^q_j \right) + \sum_{j \in T} p^q_j} \\
&\qquad\geq \Ex{\bv}{f^{\bv}(q)} - \sum_{j \in T} \Ex{\bv}{q y^{\bv,q}_j} - \sum_i \Ex{\bv}{u^{\bv,q}_i} - q \sum_{j \in M} p^q_j + \sum_{j \in T} p^q_j
\end{align*}
Notice that $\sum_{j \in T} \Ex{\bv}{q y^{\bv,q}_j} = \sum_{j \in T} p^q_j$ from the definition of $p^q_j$.  Furthermore,
\[ \sum_i \Ex{\bv}{u^{\bv,q}_i} + q \sum_{j \in M} p^q_j = \Ex{\bv}{u^{\bv,q}_i + q^2 \sum_{j \in M} y^{\bv,q}_j} = \Ex{\bv}{f^{\bv}(q^2)}\]
from the choice of $(y^{\bv,q}_j)_{j\in M}$ and $(u^{\bv,q}_i)_{i \in N}$.  We can therefore simplify our inequality to
\begin{align*}
&\Ex{\bv}{\sum_{i, S} \lambda^{i,\bv,q}_S \left( v_i(S) - v_i(S \cap T) - \sum_{j \in S} p^q_j \right) + \sum_{j \in T} p^q_j} \\
&\qquad\geq \Ex{\bv}{f^{\bv}(q)} - \Ex{\bv}{f^{\bv}(q^2)}.
\end{align*}
This completes the proof of Claim~\ref{claim:polytime-q-qsquard}.
\end{proof}

Claim~\ref{claim:polytime-q-qsquard} holds for a fixed but arbitrary choice of $q \in [0,1]$. Our algorithm chooses among all $q = 2^{-2^X}$ for $X \in \{0, \ldots, \ell\}$, where $\ell = \log\log m$, the one that maximizes $\Ex{\bv}{f^{\bv}(q)} - \Ex{\bv}{f^{\bv}(q^2)}$. Note that by this choice of $q$, we have
\begin{align*}
\Ex{\bv}{f^{\bv}(q)} - \Ex{\bv}{f^{\bv}(q^2)} & \geq \frac{1}{\ell + 1} \sum_{X = 0}^\ell \left(\Ex{\bv}{f^{\bv}(2^{-2^X})} - \Ex{\bv}{f^{\bv}((2^{-2^X})^2)}\right) \\
& = \frac{1}{\ell + 1} \Ex{\bv}{\left( f^{\bv}\left(\frac{1}{2}\right) - f^{\bv}\left(\frac{1}{m^2}\right) \right)}.
\end{align*}

As we have already established $f^{\bv}\left(\frac{1}{2}\right) \geq \frac{1}{2} \bv(\OPT(\bv))$ because one particular choice for $\lambda$ would be $\OPT(\bv)$ with probability $\frac{1}{2}$ and the empty allocation otherwise. Furthermore, $f^{\bv}\left(\frac{1}{m^2}\right) \leq \frac{1}{m} \bv(\OPT(\bv))$ because for any $\lambda \in \Gamma(\frac{1}{m^2})$, we have $\sum_{i=1}^{n} \sum_{S} \lambda^{i}_S v_i(S) \leq \sum_{i=1}^{n} \sum_{S} \lambda^{i}_S \sum_{j \in S} v_i(\{j\}) \leq \frac{1}{m^2} \sum_{j \in M} v_i(\{j\}) \leq \frac{1}{m} \bv(\OPT(\bv))$.

So, these calculations in combination with Claim~\ref{claim:polytime-q-qsquard} yield that for the value of $q$ chosen by the algorithm
\begin{align*}
&\Ex{\bv}{\sum_{i, S} \lambda^{i,\bv,q}_S \left( v_i(S \setminus T) - \sum_{j \in S} p^q_j \right) + \sum_{j \in T} p^q_j}\\
&\qquad\geq \frac{1}{\ell + 1} \left(\frac{1}{2} - \frac{1}{m}\right) \Ex{\bv}{\bv(\OPT(\bv)}.
\end{align*}

\paragraph{Estimating Expectations by Sampling.}
Our algorithm so far requires to compute expectations $\Ex{\bv}{f^{\bv}(q)}$ and $p^q_j = \Ex{\tilde{\bv}}{q y^{\tilde{\bv},q}_j}$.
%We now turn to issues of computation.  Recall that we can compute the dual variables $(y^{\bv,q}_j)_{j \in M}$ for a particular choice of $\bv$; we now wish to compute prices that approximate the properties of prices $p^q_j$.  
Given only sample access to the distributions, we will first estimate the value of $\Ex{\bv}{f^{\bv}(q)}$ for each $q$.  This is done by repeatedly sampling $\bv$ and calculating $f^{\bv}$ using the linear program (which we can solve using demand oracles).  Write $\hat{f}(q)$ for the resulting estimate.  Assuming that values are scaled to lie in $[0,1]$, Hoeffding's inequality guarantees that for each fixed $q$ with $N_1$ samples $\Pr{\lvert \hat{f}(q) - \Ex{\bv}{f^{\bv}(q)} \rvert \geq \delta} \leq \exp(-2 N_1 \delta^2)$. That is, by a union bound, with probability $1 - (\ell + 1) \exp(-2 N_1 \delta^2)$, we have $\lvert \hat{f}(q) - \Ex{\bv}{f^{\bv}(q)} \rvert < \delta$ for all $q = 2^{-2^X}$, where $X \in \{0, \ldots, \ell\}$.
%standard Chernoff bounds imply that  poly$(1/\delta,n,m)$ samples are sufficient to guarantee that $|\hat{f}(q) - \Ex{\bv}{f^{\bv}(q)}| < \delta$, for each $q$, with all but exponentially small probability.

We will then choose $q = 2^{-2^X}$ to maximize the (estimated) difference $\hat{f}(q) - \hat{f}(q^2)$.  The prices $p^q_j$ then satisfy
\begin{align*}
&\Ex{\bv}{\sum_{i, S} \lambda^{i,\bv,q}_S \left( v_i(S \setminus T) - \sum_{j \in S} p^q_j \right) + \sum_{j \in T} p^q_j} \\
&\qquad\geq \frac{1}{\ell + 1} \left( \left(\frac{1}{2} - \frac{1}{m}\right) \Ex{\bv}{\bv(\OPT(\bv)} - |\hat{f}(q) - \Ex{\bv}{f^{\bv}(q)}| - |\hat{f}(q^2) - \Ex{\bv}{f^{\bv}(q^2)}| \right)\\
&\qquad\geq \frac{1}{\ell + 1} \left(\frac{1}{2} - \frac{1}{m}\right) \Ex{\bv}{\bv(\OPT(\bv)} - 2\delta.
\end{align*}

We next compute prices $\hat{p}_j$ that estimate $p^q_j$ by sampling. For every drawn valuation profile $\bv$, we solve the dual linear program corresponding to our choice of $q$, then take $\hat{p}_j$ to be the sample mean of the observed dual values $y^{\bv,q}_j$ scaled by $q$.  Note that as values are scaled to lie in $[0,1]$ also $y^{\bv,q}_j \in [0,1]$ for all $\bv$ and $j$. Therefore, Hoeffding's inequality guarantees that for each fixed $j$ and fixed $X$ with $N_2$ samples $\Pr{\lvert \hat{p}_j - p^qs_j \rvert \geq \delta} \leq \exp(-2 N_2 \delta^2)$. That is, by a union bound, with probability $1 - m \exp(-2 N_2 \delta^2)$, we have $\lvert \hat{p}_j - p^q_j \rvert < \delta$ for all $j$ for a fixed $q$.
%, standard Chernoff bounds imply that taking poly$(1/\delta,n,m)$ samples of $p_j^{\bv}$ (again, by solving the dual linear program corresponding to our choice of $X$) is sufficient to estimate $p^X_j$ up to an additive loss of $\delta$ with all but exponentially small probability.
Conditioning upon success, we then have
\begin{align*}
&\Ex{\bv}{\sum_{i, S} \lambda^{i,\bv,q}_S \left( v_i(S \setminus T) - \sum_{j \in S} \hat{p}_j \right) + \sum_{j \in T} \hat{p}_j}\\ 
&\qquad\geq \Ex{\bv}{\sum_{i, S} \lambda^{i,\bv,q}_S \left( v_i(S \setminus T) - \sum_{j \in S} (p^q_j + \delta) \right) + \sum_{j \in T} (p^q_j-\delta)}\\
&\qquad\geq \Ex{\bv}{\sum_{i, S} \lambda^{i,\bv,q}_S \left( v_i(S \setminus T) - \sum_{j \in S} p^q_j \right) + \sum_{j \in T} p^q_j}
 - m\delta - mn\delta\\ 
&\qquad\geq \frac{1}{\ell + 1} \left(\frac{1}{2} - \frac{1}{m}\right) \Ex{\bv}{\bv(\OPT(\bv)} - (m+mn+2)\delta.
\end{align*}

Now choose $\delta = \frac{\epsilon}{m+mn+2}$. The error probability using $N_1$ and $N_2$ samples is upper-bounded by $(\ell + 1) \exp(-2 N_1 \delta^2) + m \exp(-2 N_2 \delta^2)$, so $N_1 = N_2 = \frac{1}{2 \delta^2} \ln( 2 m / \zeta ) = O(\mathrm{poly}(n,m,1/\epsilon,\log(1/\zeta)))$ guarantees it to be at most $\zeta$.
\end{proof}

%% The proof that these prices result in an $O(\log\log m)$ approximation to the optimal expected welfare will follow in a similar way as from Lemma~\ref{lem:key-lemma-bayesian}.  
The proof of Theorem~\ref{thm:upper-bound-polytime} now follows from Lemma~\ref{lem:key-lemma-bayesian-polytime} in a similar way as Theorem~\ref{thm:upper-bound} followed from Lemma~\ref{lem:key-lemma-bayesian} (see Appendix~\ref{app:proof-thm-upper-bound-polytime} for details). The advantage of Lemma~\ref{lem:key-lemma-bayesian-polytime} is that it is amenable to polynomial time computation.  If we use prices $(\hat{p}_j)_{j \in M}$, we incur an additive $\epsilon$ error for each item, for each agent.  For a total of $nm\epsilon$ additional additive error.  Taking $\epsilon$ sufficiently small yields our desired approximation in polynomial time.

%% file: revenue.tex
% !TEX root = main.tex
%% Section: Extension to Revenue
\section{An $\mathbf{O(\text{log}\, \text{log}\, m)}$ Approximation to Optimal Revenue}
\label{sec:revenue}

We next show how to extend our arguments to obtain a posted-price mechanism that achieves near-optimal revenue rather than welfare. 
We will follow the approach of Cai and Zhao~\cite{CaiZ17}. 
Fix the valuation distribution $\D = \prod_i \D_i$.  We will make an independence assumption on each distribution $\D_i$, which is that the valuations are \emph{subadditive over independent items.} 
%Recall that we use $\D = \prod_i \D_i$ to denote the value distributions. As is standard in the literature on revenue maximization, we will obtain our result under an additional independence assumption, namely that valuations are \emph{subadditive over independent items.} 
Roughly speaking, this means that for any $S$ and $T$ with $S \cap T = \emptyset$, the random variables $v_i(S)$ and $v_i(T)$ are distributed independently.  In particular, this implies $v_i(\{j\})$ is distributed independently for each item $j$.  We'll write $v_i(j) = v_i(\{j\})$ for convenience.

%\textcolor{teal}{Brendan: CZ defines this more formally using types $t_i = (t_{i1}, \dotsc, t_{im})$, where each $t_{ij}$ is distributed independently and $v_i(S) = v_i(t_i, S)$ is a deterministic function of type that depends only on $\{ t_j \colon j \in S \}$.  Is this equivalent to the condition above?  Should we also describe the formalization for clarity, but then just not use it?}

\begin{theorem}\label{thm:revenue}
When buyers have subadditive valuations over independent items, there exists a simple, deterministic, and DSIC mechanism that achieves an $\Omega(1/\alpha)$ approximation to the optimal BIC revenue where $\alpha = O(\log \log m)$.  
\end{theorem}

%Our mechanism is from one of the following two classes:
Just like the mechanism of Cai and Zhao~\cite{CaiZ17}, our mechanism will be from one of the following two classes:
\begin{enumerate}
\item Rationed sequential posted-price mechanism (RSPM): The buyers are approached in a fixed order.  For each buyer, each item is assigned a static and potentially personalized posted price.  Each buyer can purchase at most a single item at its listed price.
\item Anonymous sequential posted-price with entry fee mechanism (ASPE): Each item is assigned a static anonymous posted price.  The buyers are approached in a fixed order, and each buyer faces an entry fee that can depend on the set of items that have not yet been sold when they arrive.  If the buyer pays the entry fee, they can purchase any set of items at their posted prices.
\end{enumerate}

We will write $\textsc{PostRev}$ for the optimal revenue attainable using a RSPM, $\textsc{APostEnRev}$ for the optimal revenue attainable using an ASPE, and $\textsc{Rev}$ for the optimal BIC revenue.

%Our proof of Theorem~\ref{thm:revenue} follows the same general strategy as the proof of Cai and Zhao~\cite{CaiZ17}, but improves the approximation guarantee from $O(\log m)$ to $O(\log \log m)$ by using our new method for item pricing.

In Section~\ref{sec:core-decomposition} and Section~\ref{sec:core-within-the-core} we describe the high-level approach of Cai and Zhao~\cite{CaiZ17} and the key facts of their construction that we will reuse. In Section~\ref{sec:core-within-the-core} we also state our key lemma, Lemma~\ref{lem:pp-with-entry-fees}, and show how it implies the improved bound in Theorem~\ref{thm:revenue}. We prove Lemma~\ref{lem:pp-with-entry-fees} in Section~\ref{sec:revenue-argument}. 
%and our argument to establish Theorem~\ref{thm:revenue}. The basic construction follows \cite{CaiZ17}, but notes that the parts of their analysis that we re-use are actually independent of the specific way they set the prices. The improved bound follows from Lemma~\ref{lem:pp-with-entry-fees} which we prove in Section~\ref{sec:revenue-argument}.

\subsection{A First Core-Tail Decomposition and a Benchmark}
\label{sec:core-decomposition}

%\textcolor{teal}{Paul: Say that CZ give us a core-tail decomposition (via the $\beta_{i,j}$ below) together with an interim (but not necessarily ex -post feasible) mechanism (the $\sigma$ below) that provides a constant factor approximation to the optimal revenue. Rather than approximating the revenue of this mechanism, we approximate its welfare.}

In this section we summarize a particular core-tail decomposition due to Cai and Zhao~\cite{CaiZ17}.  First some notation.  We'll use $\mathbf{\sigma}$ to describe an interim allocation rule of mechanism, where $\sigma_S^{i}(v_i)$ is the probability that agent $i$ is allocated set $S$ when agent $i$ has valuation $v_i$.  We'll write $\pi_{j}^{i}(v_i) = \sum_{S \ni j}\sigma_S^{i}(v_i)$ for the probability that agent $i$ is allocated item $j$ under valuation $v_i$.  
%% PAUL: Not sure we need this
% In general, an interim allocation rule need not be ex post feasible.

%Fix a mechanism with interim allocation rule $\mathbf{\sigma}$, where $\sigma_S^{i}(v_i)$ is the probability that agent $i$ is allocated set $S$ when agent $i$ has valuation $v_i$.  We'll write $\pi_{j}^{i}(v_i) = \sum_{S \ni j}\sigma_S^{i}(v_i)$ for the probability that agent $i$ is allocated item $j$ under valuation $v_i$.  We will think of $\sigma$ as describing a revenue-optimal (or approximately revenue-optimal) mechanism.

Fix personalized item thresholds $\beta_{ij} \geq 0$ for all $i$ and $j$, whose values will be chosen later.  We now describe what is meant by the core.  For each buyer $i$, define a threshold $c_i$ as follows:
\[ c_i = \inf\{ x \geq 0 \colon \sum_j \Prr{v_i}{v_i(j) \geq \beta_{ij}+x} \leq 1/2 \}. \]
Let $C_i(v_i) = \{ j \colon v_i(j) < \beta_{ij} + c_i \}$.  That is, $C_i(v_i)$ contains all items $j$ for which agent $i$ does not have too large a value for item $j$ as a singleton, relative to threshold $\beta_{ij}$.

Fix a mechanism $M$ with interim allocation rule $\mathbf{\sigma}$.  We can now define the core of a mechanism as follows:
\[ \textsc{Core}(M,\mathbf{\beta}) = \Ex{\bv}{\sum_{i,S} \sigma_S^i(v_i) v_i(S \cap C_i(v_i))}. \]
That is, the core with respect to $M$ and $\beta$ is the total welfare generated by $M$, excluding any item $j$ assigned to a buyer $i$ that has too large a value for it, relative to $\beta_{ij}$.  

%Recall that the core describes welfare for a restricted valuation, which cannot obtain value from certain high-value items.  
We can interpret this more explicitly as welfare under a valuation transformation.  Define $v'_i(S) = v_i(S \cap C_i(v_i))$.  Note that with this definition, omitting the dependence of $\sigma$ and $\bv'$ on $\bv$ from the notation, $\textsc{Core}(M,\beta) = \Ex{}{\sum_{i,S} \sigma_S^i v'_i(S).}$ An important observation due to Cai and Zhao \cite{CaiZ17} is the following:
%Then CZ show that $v'_i$ is subadditive over independent items.  

\begin{lemma}[Cai and Zhao~\cite{CaiZ17}]
Valuations $\bv'$ are subadditive over independent items.
\end{lemma}

The following lemma summarizes the implications of a construction due to Cai and Zhao~\cite{CaiZ17}

\begin{lemma}[Cai and Zhao~\cite{CaiZ17}]\label{lem:cz}
%\textcolor{teal}{Paul: Todo: Dependence of $d_1,d_2$ on $b$.} %Note: I have added 2. below.
Fix an arbitrary constant $b \in (0,1)$. Then there exists a mechanism $M$ with interim allocation rule $\mathbf{\sigma}$ and corresponding item allocation rules $\mathbf{\pi}$ and personalized item thresholds $\beta_{ij}$ for all $i$ and $j$ such that the following are true: %, and constants $d_1, d_2 \geq 1$ such that the following are true:
\begin{enumerate}
%% Paul: This is Theorem 2, Lemma 13, Lemma 14, and Lemma 15 in CZ
%\item $\textsc{Rev} \leq d_1 \cdot \textsc{PostRev} + d_2 \cdot \textsc{Core}(M,\beta)$, 
\item $\textsc{Rev} \leq \left(\frac{8}{1-b} + 12 \right) \cdot \textsc{PostRev} + 4 \cdot \textsc{Core}(M,\beta)$, 
%% Paul: This is Lemma 18 in CZ
\item $\sum_{i=1}^{n} \frac{c_i}{2} \leq \frac{2}{1-b} \cdot \textsc{PostRev}$,
%% Paul: This is Lemma 5 in CZ
\item $\sum_{k \neq i}\Prr{v_k}{v_k(j) \geq \beta_{kj}} \leq b$, 
\item $\Ex{v_i}{\pi^i_j(v_i)} \leq \Prr{v_i}{v_i(j) \geq \beta_{ij}}/b$. 
\end{enumerate}
\end{lemma}

From now on we'll fix $b$ and consider the corresponding mechanism $M$ and thresholds $\beta$ from Lemma~\ref{lem:cz}, and we'll simply write $\textsc{Core}$ to mean $\textsc{Core}(M,\beta)$.  Given Lemma~\ref{lem:cz}, what remains in order to prove Theorem~\ref{thm:revenue} is to argue that $\textsc{Core}$ can be approximated (up to a factor of $O(\log\log m)$) by either $\textsc{PostRev}$ or $\textsc{APostEnRev}$.

\subsection{The Core Within the Core, Key Lemma, and How it Implies the Bound}
\label{sec:core-within-the-core}

The \textsc{Core} as defined above is just the welfare for a specific allocation under transformed but still subadditive valuation functions (namely the $v'_i$'s). A key idea in the literature on simple, near-optimal posted price mechanisms for revenue is to turn posted-price mechanisms that achieve some approximation guarantee for welfare into mechanisms that achieve the same (up to constant factors) approximation guarantee for revenue by augmenting the mechanism with entry fees. The idea is that if each buyer's surplus is sufficiently concentrated, then that surplus can be extracted as revenue using entry fees.  One way to show concentration is to argue that no single item's contribution to the surplus is too large (i.e., a Lipschitz condition), but as it turns out individual items can contribute significantly to surplus under valuations $v'_i$.

Cai and Zhao~\cite{CaiZ17} therefore invoke a second core restriction (called $\widehat{\textsc{CORE}}$ in their paper), resulting in a further restricted valuation $\hat{v}_i$ for each agent.  This further restricted valuation has a sufficiently small Lipschitz constant. % that we are able establish the concentration bound necessary for our analysis.  
%This part of the analysis is similar to the construction of $\widehat{\textsc{CORE}}$ in \cite{CaiZ17}.

To define this second restriction, consider a set of item prices $p_j$ (we will fix these later). Then, for each agent $i$, define
\[ \tau_i = \inf\{ x \colon \sum_j \Pr{v_i(j) \geq \max\{\beta_{ij}, p_j + x\}} \leq 1/2 \}. \]
Write $Y_i(v_i) = \{ j \colon v_i(j) < p_j + \tau_i \}$.  
%Write $Y_i(v_i) = \{ j \colon v'_i(j) < p_j + \tau_i \}$. 

With these definitions in place, we can formalize what we meant by a further restriction of each agent's valuation $v_i$.  Define $\hat{v}_i(S) = v_i(S \cap Y_i(v_i))$.  So $\hat{v}$ is like $v$, but with any ``very high-valued individual items'' removed from the valuation.  Unlike the initial core decomposition, now the high-valuedness is with respect to prices $p_j$ rather than the thresholds $\beta_{ij}$.

For each agent $i$, set $S$, and valuation $v_i$, set $\mu_i(v_i, S) = \max_{S' \subseteq S}( \hat{v}_i(S') - \sum_{j \in S'}p_j )$. 
%\textcolor{teal}{Brendan: In retrospect, $\mu_i$ is bad notation to use here, since we already use it for the adversary fractional strategy.  Will need to change.} 
This is the surplus enjoyed by an agent with valuation $\hat{v}_i$ when $S$ is the set of available items, priced according to $p_j$. 
(We use the notation $\mu_i(v_i,S)$ for consistency with \cite{CaiZ17}. It's not related to $\mu_T$ as it appears in the zero-sum games.)

The following lemmata summarize key properties of the core-within-the-core that are shown in \cite{CaiZ17} and that we will re-use in our analysis.

%% Paul: This is Lemma 25 in CZ
\begin{lemma}[Cai and Zhao~\cite{CaiZ17}]
%\textcolor{teal}{Paul: Also monotone and subadditive, do we need these properties?}
For any choice of item prices $p_j$ and subadditive functions $v_i$, the surplus $\mu_i$ is monotone, subadditive, and $\tau_i$-Lipschitz.
\end{lemma}

%% Paul: This is Lemma 21 in CZ
\begin{lemma}[Cai and Zhao~\cite{CaiZ17}]\label{lem:cz-lemma21}
%\textcolor{teal}{Paul: Say that $b$ comes from Lemma~\ref{lem:cz}}
For any choice of item prices $p_j$ and subadditive functions $v_i$,
\[
	\sum_{i=1}^{n} \sum_{j \in M} \max\{\beta_{ij}, p_j+\tau_i\} \cdot \Prr{v_i}{v_i(j) \geq \max\{\beta_{i,j}, p_j+\tau_i\}} \leq \frac{2}{b(1-b)} \cdot \textsc{PostRev}.
\]
\end{lemma}

%The following should hold similarly to Lemma 22 in CZ, basically without change (here $\textsc{PostRev}$ is the maximum revenue of a posted-price mechanism where each buyer can buy at most one item):
\begin{lemma}[Cai and Zhao~\cite{CaiZ17}] \label{lem:tau-and-postrev}
%\textcolor{teal}{Paul: Say that $b$ comes from Lemma~\ref{lem:cz}}
For any choice of item prices $p_j$ and subadditive functions $v_i$,
$\sum_i \tau_i \leq \frac{4}{1-b} \cdot \textsc{PostRev}$.
\end{lemma}
%\textcolor{teal}{Thomas: I checked the proof of Lemma 22 in CZ. It seems everything holds without any further change. I don't see an immediate need to copy everything.}

The crux and key innovation of our analysis is now the following lemma, which establishes the existence of an ASPE with an appropriate approximation guarantee.

Our mechanism is actually from the same class of ASPE mechanisms as the mechanism of Cai and Zhao~\cite{CaiZ17}, but uses a different set of item prices. This class of ASPE mechanisms, less call them \emph{median ASPE}, is parametrized by a set of of item prices $p_j$ and the mechanisms within that class proceed as follows.
%We'll define a sequential mechanism with entry fees.  
The agents are approached sequentially in a fixed order.  We write $S_i(\bv)$ for the set of items still available when agent $i$ is approached.  (Note that $S_i(\bv)$ only depends on the entries of $\bv$ corresponding to agents that arrived before $i$.)  Agent $i$ faces an entry fee equal to the median (over randomness in $v_i$) of $\mu_i(v_i, S_i(\bv))$.  If agent $i$ pays the entry fee, she can then purchase any desired subset of items at prices $p_j$.

\begin{lemma}\label{lem:pp-with-entry-fees}
%\textcolor{teal}{Paul: Say that $b$ comes from Lemma~\ref{lem:cz}}
There is a set of item prices $p_j$ such that the median ASPE with these prices achieves expected revenue at least
\[
\textsc{APostEnRev} \geq \frac{1}{4\alpha} \cdot \textsc{Core} - \frac{1 + 6b}{2b (1 - b)} \cdot \textsc{PostRev}.
%APostEnRev \geq \frac{1}{4\alpha} \textsc{CORE} - \frac{1 + 6b}{2b (1 - b)} \cdot \textsc{PostRev}.
\]
with $\alpha = O(\log \log m)$.
\end{lemma}

Before we prove this lemma let's first see how it implies the result that we want to prove.

\begin{proof}[Proof of Theorem~\ref{thm:revenue}]
Let $b \in (0,1)$ and $c>0$ be arbitrary constants. If $\textsc{PostRev} \geq c/\alpha \cdot \textsc{Core}$ then by Lemma~\ref{lem:cz},
\begin{align*} 
\textsc{PostRev} \geq \frac{1}{\frac{8}{1-b} + 12 + 4 \cdot \frac{\alpha}{c}} \cdot \textsc{Rev}. 
\end{align*}
Otherwise, Lemma~\ref{lem:pp-with-entry-fees} combined with Lemma~\ref{lem:cz} shows that, 
\begin{align*}
\textsc{APostEnRev} 
&\geq \left(\frac{1}{4\alpha} - \frac{1+6b}{2b(1-b)} \cdot \frac{c}{\alpha}\right) \cdot \textsc{Core} 
\geq \frac{\frac{1}{4\alpha} - \frac{1+6b}{2b(1-b)} \cdot \frac{c}{\alpha}}{\left(\frac{8}{1-b}+12\right) \cdot \frac{c}{\alpha} + 4} \cdot \textsc{Rev}.
\end{align*}
Now let's choose $c$ in dependence of $b$ such that
\begin{align*}
\frac{1}{\frac{8}{1-b} + 12 + 4 \cdot \frac{\alpha}{c}} = \frac{\frac{1}{4\alpha} - \frac{1+6b}{2b(1-b)} \cdot \frac{c}{\alpha}}{\left(\frac{8}{1-b}+12\right) \cdot \frac{c}{\alpha} + 4} 
%\Leftrightarrow c = \frac{1}{4} \cdot \frac{2b(1-b)}{8b-2b^2+1}.
\end{align*}
This gives $c(b) = \frac{1}{4} \cdot \frac{2b(1-b)}{8b-2b^2+1}$. 
Then,
\begin{align*}
\max\{\textsc{PostRev}, \textsc{APostEnRev}\} \geq \frac{1}{\frac{8}{1-b} + 12 + 4 \cdot \frac{\alpha}{c(b)}} \cdot \textsc{Rev} = \frac{1}{\frac{8}{1-b}+12+16 \alpha \cdot \frac{8b-2b^2+1}{2b(1-b)}} \cdot \textsc{Rev}.
\end{align*} 
Choosing, e.g., $b = 1/2$ we obtain the claim. 
%while Lemma~\ref{lem:cz} shows that $O(1) \cdot \textsc{Core} \geq Rev$.
\end{proof}

\subsection{Proof of Key Lemma, Pricing Beyond Additive Supporting Functions}
\label{sec:revenue-argument}

It remains to show Lemma~\ref{lem:pp-with-entry-fees}, which boils down to finding appropriate item prices. Cai and Zhao \cite{CaiZ17} use the ``usual'' approach of approximating subadditive valuations with XOS valuations \cite{BhawalkarR11,Dobzinski07}, and using this to set item prices. We show how to leverage our novel approach to item pricing to improve the bound.

\subsubsection{Crafting the Item Prices}

We would like to find item prices that guide the allocation toward one that approximates the allocation that defines $\textsc{Core} = \Ex{}{\sum_{i,S} \sigma_S^i v'_i(S)}$.  
To this end we will extend Lemma~\ref{lem:key-lemma-bayesian} (our key lemma in the incomplete information case) to allow arbitrary constraints on the marginal probability of allocating each item to each agent.  Given $\mathbf{z} = \{z_{ij}\}$, where each $z_{ij}$ is a function with $z_{ij}(v_i) \in [0,1]$ for each $i$, $j$, and $v_i$, we'll write 
\[ 
	\Lambda(\mathbf{z}) = \left\{ \{\lambda^\bv\}_\bv \colon \Ex{\bv_{-i}}{\sum_{S \ni j} \lambda^{i,\bv}_{S}} \leq z_{ij}(v_i) \ \forall\ i,j,v_i \right\}. 
\]
That is, $\Lambda(\mathbf{z})$ is the set of all collection of $\lambda^\bv$'s satisfying the upper bounds described by $\mathbf{z}$.

Analogous to Lemma~\ref{lem:key-lemma-bayesian} we can now show the following lemma. The proof follows the proof of Lemma~\ref{lem:key-lemma-bayesian} with minor changes. See Appendix~\ref{app:proof-key-lemma-revenue}.

\begin{lemma}\label{lem:key-lemma-revenue}
For every independent probability distribution $\D = \prod_i \D_i$ over subadditive valuation functions, and $z_{ij}(v_i) \in [0,1]$ for each $i$ and $j$ and $v_i$, there exist prices $p_j$ for $j \in M$ and probability distributions $\lambda^{i,\bv}$ over $S \subseteq M$ for all $i$ and $\bv$ such that $\lambda \in \Lambda(\mathbf{z})$ and, for all $T \subseteq M$,
%such that for all $T \subseteq U$ there is $S \subseteq U\setminus T$ for which
\[
	\sum_{j \in T} p_j + \Ex{\bv}{\sum_{i=1}^{n} \sum_{S} \lambda^{i,\bv}_S \left(v'_i(S \setminus T) - \sum_{j \in S} p_j\right)} \geq \frac{1}{\alpha} \cdot \Ex{\bv}{ \max_{\overline{\lambda} \in \Lambda(\mathbf{z})} \sum_{i,S} \overline{\lambda}^{i,\bv}_S v'_i(S)}, 
	%\sum_{j \in T} p_j + \frac{1}{2} \left(v_i(S) - \sum_{j \in S} p_j\right) \geq \frac{v_i(U)}{4 \alpha}, 
\]
where $\alpha \in O(\log \log m)$.
\end{lemma}

Invoke this lemma with $z_{ij}(v_i) = \pi_j^i(v_i)$, resulting in some $p_j$ and $\lambda$.  Note then that the RHS of Lemma~\ref{lem:key-lemma-revenue}, $\Ex{\bv}{ \max_{\overline{\lambda} \in \Lambda(\mathbf{z})} \sum_{i,S} \overline{\lambda}^{i,\bv}_S v'_i(S)}$, is at least $\textsc{Core}$.  

We'd now like to claim that if we replace $v'_i$ by $\hat{v}_i$ in Lemma~\ref{lem:key-lemma-revenue}, this does not change the welfare bound by too much.

\begin{lemma}
\label{lemma:welfare-with-without-hat}
%\textcolor{teal}{Paul: Say that $b$ comes from Lemma~\ref{lem:cz}}
For any $\{\lambda^\bv\}_\bv \in \Lambda(\mathbf{z})$ with $z_{ij}(v_i) = \pi_j^i(v_i)$ and for all $T \subseteq M$, 
%\[ \Ex{\bv}{\sum_{i=1}^{n} \sum_{S} \lambda^{i,\bv}_S v'_i(S \setminus T)} - \Ex{\bv}{\sum_{i=1}^{n} \sum_{S} \lambda^{i,\bv}_S \hat{v}_i(S \setminus T)} \leq \frac{40}{3} \cdot PostRev \]
\[ \Ex{\bv}{\sum_{i=1}^{n} \sum_{S} \lambda^{i,\bv}_S v'_i(S \setminus T)} - \Ex{\bv}{\sum_{i=1}^{n} \sum_{S} \lambda^{i,\bv}_S \hat{v}_i(S \setminus T)} \leq \frac{2 (1+b)}{b (1-b)} \cdot \textsc{PostRev} \]
\end{lemma}

\begin{proof}
Our goal is to bound, for every $i$ and every fixed $v_i$, 
\[
\Ex{\bv_{-i}}{\sum_{S} \lambda^{i,\bv}_S v'_i(S \setminus T)} - \Ex{\bv_{-i}}{\sum_{S} \lambda^{i,\bv}_S \hat{v}_i(S \setminus T)} = \Ex{\bv_{-i}}{\sum_{S} \lambda^{i,\bv}_S \left(v'_i(S \setminus T) - \hat{v}_i(S \setminus T) \right)}.
\]
%\textcolor{teal}{Paul: We only need subadditivity of $v'_i$ and monontonicity of $\hat{v}_i$ for this, right?}
First, we upper bound $v'_i(S \setminus T) - \hat{v}_i(S \setminus T)$ using the definition of $\hat{v}_i$, that $\hat{v}_i$ is monotone, and that $v'_i$ is subadditive as follows:
\[
v'_i(S \setminus T) - \hat{v}_i(S \setminus T) = v'_i(S \setminus T) - \hat{v}_i((S \setminus T) \cap Y_i(v_i)) \leq \sum_{j \in S \setminus (T \cup Y_i(v_i))} v'_i(j) \leq \sum_{j \in S \setminus Y_i(v_i)} v'_i(j).
\]
Note that $j \not\in Y_i(v_i)$ only if $v_i(j) \geq p_j + \tau_i$. Furthermore, $v'_i(j) \leq \beta_{ij} + c_j$ for all $j$. 
But it may also happen that $v_i(j) \leq \beta_{ij}$ and hence $v'_i(j) \leq \beta_{ij}$. 
%% Old version below
%Note that $j \not\in Y_i(v_i)$ only if $v_i(\{ j \}) \geq p_j + \tau_i$. Furthermore, $v_i(\{ j \}) \leq \beta_{ij} + c_j$ for all $v_i$, recalling that $v_i(S) \leftarrow v_i(S \cap C_i(v_i))$. 
%But it may also happen that $v_i(\{ j \}) \leq \beta_{ij}$. 
This lets us rewrite the above sum to
\[
\sum_{j \in S \setminus Y_i(v_i)} v'_i(\{ j \}) \leq \sum_{j \in S} \left(\beta_{ij} \mathbf{1}_{v_i(j) \geq p_j + \tau_i} + c_j \mathbf{1}_{v_i(j) \geq \max\{p_j + \tau_i, \beta_{ij}\}}\right).
\]

%\textcolor{teal}{Paul: The third line had a sum over $j\in M\setminus T$, but I think it should be $j \in M$} 
Therefore, using that $\{\lambda^\bv\}_\bv \in \Lambda(\mathbf{z})$ with $z_{ij}(v_i) = \pi_j^i(v_i)$, we have
\begin{align*}
&\Ex{\bv_{-i}}{\sum_{S} \lambda^{i,\bv}_S \left(v'_i(S \setminus T) - \hat{v}_i(S \setminus T) \right)} \\
&\qquad\qquad\leq \Ex{\bv_{-i}}{\sum_{S} \lambda^{i,\bv}_S \sum_{j \in S} \left(\beta_{ij} \mathbf{1}_{v_i(j) \geq p_j + \tau_i} + c_j \mathbf{1}_{v_i(j) \geq \max\{p_j + \tau_i, \beta_{ij}\}}\right)} \\
&\qquad\qquad= \Ex{\bv_{-i}}{\sum_{j \in M} \sum_{S \ni j} \lambda^{i,\bv}_S \left( \beta_{ij} \mathbf{1}_{v_i(j) \geq p_j + \tau_i} + c_j \mathbf{1}_{v_i(j) \geq \max\{p_j + \tau_i, \beta_{ij}\}} \right)} \\
&\qquad\qquad= \sum_{j \in M} \Ex{\bv_{-i}}{\sum_{S \ni j} \lambda^{i,\bv}_S} \left( \beta_{ij} \mathbf{1}_{v_i(j) \geq p_j + \tau_i} + c_j \mathbf{1}_{v_i(j) \geq \max\{p_j + \tau_i, \beta_{ij}\}} \right) \\
&\qquad\qquad\leq \sum_{j \in M} \pi_j^i(v_i) \left( \beta_{ij} \mathbf{1}_{v_i(j) \geq p_j + \tau_i} + c_j \mathbf{1}_{v_i(j) \geq \max\{p_j + \tau_i, \beta_{ij}\}} \right).
\end{align*}
Taking the expectation over $v_i$ and the sum over all $i$, we obtain
\begin{align}
& \Ex{\bv}{\sum_{i=1}^{n} \sum_{S} \lambda^{i,\bv}_S v'_i(S \setminus T)} - \Ex{\bv}{\sum_{i=1}^{n} \sum_{S} \lambda^{i,\bv}_S \hat{v}_i(S \setminus T)} \notag \\
&\qquad\qquad= \sum_{i=1}^{n} \Ex{\bv}{\sum_{S} \lambda^{i,\bv}_S \left(v'_i(S \setminus T) - \hat{v}_i(S \setminus T) \right)} \notag \\
&\qquad\qquad\leq \sum_{i=1}^{n} \Ex{v_i}{\sum_{j \in M} \pi_j^i(v_i) \left( \beta_{ij} \mathbf{1}_{v_i(j) \geq p_j + \tau_i} + c_j \mathbf{1}_{v_i(j) \geq \max\{p_j + \tau_i, \beta_{ij}\}} \right)}. \label{eq:upper-bound}
\end{align}

%\textcolor{teal}{Paul: The rest is new}
It remains to upper bound the RHS of \eqref{eq:upper-bound} by $\frac{2 (1+b)}{b (1-b)} \cdot \textsc{PostRev}$. Let $A_i = \{j \;\vert\; \beta_{ij} \leq p_j + \tau_i\}.$ We first bound $\sum_{i=1}^{n} \Ex{v_i}{\sum_{j \in M} \pi_j^i(v_i) \beta_{ij} \mathbf{1}_{v_i(j) \geq p_j + \tau_i}}$:
\begin{align}
\sum_{i=1}^{n} \Ex{v_i}{\sum_{j \in M} \pi_j^i(v_i) \beta_{ij} \mathbf{1}_{v_i(j) \geq p_j + \tau_i}} 
&\leq \sum_{i=1}^{n} \sum_{j \in A_i}  \beta_{ij} \Ex{v_i}{\mathbf{1}_{v_i(j) \geq p_j + \tau_i}}  + \sum_{i=1}^{n} \sum_{j \not\in A_i} \beta_{ij} \Ex{v_i}{\pi_j^i(v_i)} \notag\\
&= \sum_{i=1}^{n} \sum_{j \in A_i}  \beta_{ij} \Prr{v_i}{v_i(j) \geq p_j + \tau_i} + \sum_{i=1}^{n} \sum_{j \not\in A_i} \beta_{ij} \Ex{v_i}{\pi_j^i(v_i)} \notag\\
&\leq \sum_{i=1}^{n} \sum_{j \in A_i}  \beta_{ij} \Prr{v_i}{v_i(j) \geq p_j + \tau_i} + \sum_{i=1}^{n} \sum_{j \not\in A_i} \beta_{ij} \frac{\Prr{v_i}{v_i(j) \geq \beta_{ij}}}{b} \notag\\
&\leq \frac{1}{b} \sum_{i=1}^{n} \sum_{j \in M} \max\{\beta_{ij},p_j+\tau_i\} \Prr{}{v_i(j) \geq \max\{\beta_{ij},p_j+\tau_i \}} \notag\\
&\leq \frac{2}{b(1-b)} \cdot \textsc{PostRev}. \label{eq:first-term-on-RHS}
\end{align}

The first inequality uses that both $\pi_j^i(v_i) \leq 1$ and $\mathbf{1}_{v_i(j) \geq p_j + \tau_i} \leq 1$. The second inequality uses Lemma~\ref{lem:cz}. The third inequality uses the definition of $A_i$ and that $b_i \in (0,1)$ and hence $1/b > 1$. The final inequality is Lemma~\ref{lem:cz-lemma21}.

We next bound $\sum_{i=1}^{n} \Ex{v_i}{\sum_{j \in M} \pi_j^i(v_i) c_j \mathbf{1}_{v_i(j) \geq \max\{p_j + \tau_i, \beta_{ij}\}}}$:
\begin{align}
\sum_{i=1}^{n} \Ex{v_i}{\sum_{j \in M} \pi_j^i(v_i) c_j \mathbf{1}_{v_i(j) \geq \max\{p_j + \tau_i, \beta_{ij}\}}} \notag
&\leq \sum_{i=1}^{n} c_i \sum_{j \in M} \Ex{v_i}{\mathbf{1}_{v_i(j) \geq \max\{p_j + \tau_i, \beta_{ij}\}}} \notag\\
&= \sum_{i=1}^{n} c_i \sum_{j \in M} \Prr{v_i}{v_i(j) \geq \max\{p_j + \tau_i, \beta_{ij}\}} \notag\\
&\leq \sum_{i=1}^{n} \frac{c_i}{2} \notag\\
&\leq \frac{2}{1-b} \cdot \textsc{PostRev}. \label{eq:second-term-on-RHS}
\end{align}

The first inequality uses that $\pi_j^i(v_i) \leq 1$. The second inequality follows from the definition of $\tau_i$. The final inequality uses Lemma~\ref{lem:cz}.

Combining Inequality (\ref{eq:upper-bound}) with Inequalities (\ref{eq:first-term-on-RHS}) and (\ref{eq:second-term-on-RHS}) shows the claim.
\end{proof}

Similarly to the way we prove Theorem~\ref{thm:upper-bound} from Lemma~\ref{lem:key-lemma-bayesian} we can now prove:

\begin{lemma}\label{lem:lower-bound-on-core-utilities}
%\textcolor{teal}{Paul: Say that $b$ comes from Lemma~\ref{lem:cz}} 
Using the prices $p_j$ that result from invoking Lemma~\ref{lem:key-lemma-revenue} with $z_{ij}(v_i) = \pi_j^i(v_i)$ we have
\[
\sum_{i=1}^{n} \Ex{\bv}{\mu_i(v_i, S_i(\bv))} \geq \frac{1}{\alpha} \textsc{Core} - \Ex{\bv} {\sum_{j \in
\SOLD(\bv)} p_j} - \frac{2 (1+b)}{b (1-b)} \cdot \textsc{PostRev}.
\]
\end{lemma}

\begin{proof}
Let $\{\lambda^\bv\}_\bv$ be the collection of distributions corresponding to prices $p_j$. 
Consider an arbitrary buyer $i$ and an auxiliary valuation profile $\bar{\bv}$. In order to obtain a lower bound on $\mu_i$ when the buyer has valuation $v_i$, consider drawing a set $S$ from $\lambda^{i,(v_i, \bar{\bv}_{-i})}$ and buying $S \setminus \SOLD(\bar{v}_i, \bv_{-i}) \subseteq S_i(\bv)$. We draw $\bar{\bv}$ from $\D$ and take the expectation over $\bar{\bv}$. Since $\bar{v}_i$ and $\bar{\bv}_{-i}$ are independent, this yields  
\[
 \mu_i(v_i, S_i(\bv)) \geq
 \Ex{\bar{\bv}}{\sum_{S} \lambda^{i,(v_i, \bar{\bv}_{-i})}_S \left(\hat{v}_i(S
\setminus \SOLD(\bar{v}_i,\bv_{-i})) - \sum_{j \in S} p_j\right)}.
\]

Next, we take expectations over $\bv$ on both sides. As $\bv_{-i}$ and $\bar{\bv}_{-i}$ are identically and independently distributed, we obtain
\begin{align*}
 \Ex{\bv}{\mu_i(v_i, S_i(\bv))}
 &\geq \Ex{\bv,\bar{\bv}}{\sum_{S} \lambda^{i,(v_i, \bar{\bv}_{-i})}_S \left(\hat{v}_i(S
\setminus \SOLD(\bar{v}_i,\bv_{-i})) - \sum_{j \in S} p_j\right)}\\ %}
 &= \Ex{\bv,\bar{\bv}}{\sum_{S} \lambda^{i,(v_i, \bv_{-i})}_S \left(\hat{v}_i(S
\setminus \SOLD(\bar{v}_i,\bar{\bv}_{-i})) - \sum_{j \in S} p_j\right)}.
\end{align*}

Summing this inequality over all buyers $i$ gives
\[
\sum_{i=1}^{n} \Ex{\bv}{\mu_i(v_i, S_i(\bv))}
\geq
\sum_{i=1}^{n}
\Ex{\bv,\bar{\bv}}{\sum_{S} \lambda^{i,\bv}_S \left(\hat{v}_i(S \setminus
\SOLD(\bar{\bv})) - \sum_{j \in S} p_j\right)}.
\]

Applying Lemma~\ref{lemma:welfare-with-without-hat} with $T = \SOLD(\bar{\bv})$ and taking an expectation over $\bar{\bv}$, we also have
\[
\Ex{\bv, \bar{\bv}}{\sum_{i=1}^{n} \sum_{S} \lambda^{i,\bv}_S v'_i(S \setminus \SOLD(\bar{\bv}))} - \Ex{\bv, \bar{\bv}}{\sum_{i=1}^{n} \sum_{S} \lambda^{i,\bv}_S \hat{v}_i(S \setminus \SOLD(\bar{\bv}))} \leq \frac{2 (1+b)}{b (1-b)} \cdot \textsc{PostRev}.
\]
So, in combination, by linearity of expectation
\[
\sum_{i=1}^{n} \Ex{\bv}{\mu_i(v_i, S_i(\bv))}
\geq
\sum_{i=1}^{n}
\Ex{\bv,\bar{\bv}}{\sum_{S} \lambda^{i,\bv}_S \left(v'_i(S \setminus
\SOLD(\bar{\bv})) - \sum_{j \in S} p_j\right)} - \frac{2 (1+b)}{b (1-b)} \cdot \textsc{PostRev}.
\]

We now apply Lemma~\ref{lem:key-lemma-revenue}. It gives us that for any $\bar{\bv}$,
\[
\Ex{\bv}{\sum_{S} \lambda^{i,\bv}_S \left(v'_i(S \setminus
\SOLD(\bar{\bv})) - \sum_{j \in S} p_j\right)} \geq \frac{1}{\alpha} \cdot \textsc{CORE} - \sum_{j \in \SOLD(\bar{\bv})} p_j.
\]

So, in combination, also taking the expectation over $\bar{\bv}$ here, we obtain
\[
\sum_{i=1}^{n} \Ex{\bv}{\mu_i(v_i, S_i(\bv))} \geq \frac{1}{\alpha} \cdot \textsc{CORE} - \Ex{\bar{\bv}}{\sum_{j \in
\SOLD(\bar{\bv})} p_j} - \frac{2 (1+b)}{b (1-b)} \cdot \textsc{PostRev}. 
\]
Using the fact that $\bv$ and $\bar{\bv}$ are identically distributed, the claim follows.
\end{proof}

\subsubsection{Analysis of Entry Fees}

We also need the following result from \cite{CaiZ17} concerning the revenue collected by a median ASPE through the entry fees, which expliots that the $\mu_i$ are $\tau_i$-Lipschitz. Recall that we used $S_i(\bv)$ to denote the set of items that are still available when agent $i$ is approached. Denote the median of $\mu_i(v_i,S_i(\bv))$ by $\delta_i(S_i(\bv))$.

\begin{lemma}[Cai and Zhao~\cite{CaiZ17}]\label{lem:entry-fees}
For any choice of item prices $p_j$ and subadditive valuations $v_i$, the expected revenue that the median ASPE that uses these prices collects through the entry fees is at least
\[
	\frac{1}{4} \sum_{i=1}^{n} \Ex{\bv}{\mu_i(v_i, S_i(\bv))} - \frac{5}{8} \sum_{i=1}^{n} \tau_i.
\]
\end{lemma}

\subsubsection{Putting Everything Together}

We are now ready to complete the proof of Lemma~\ref{lem:pp-with-entry-fees}.

\begin{proof}[Proof of Lemma~\ref{lem:pp-with-entry-fees}]
%\textcolor{teal}{Paul: We should say that we are using Lemma~\ref{lem:entry-fees} and Lemma~\ref{lem:tau-and-postrev} to bound the revenue from the entry fees.}

The expected revenue from the posted prices by our median ASPE is, by definition, $\Ex{\bv}{\sum_{j \in \SOLD(\bv)} p_j}$. Combining this with the lower bound on the expected revenue collected from the entry fees from Lemma~\ref{lem:entry-fees} yields the following lower bound on the expected revenue 
\begin{align*}
\textsc{APostEnRev} \geq \frac{1}{4} \sum_{i=1}^{n} \Ex{\bv}{\mu_i(v_i, S_i(\bv))} - \frac{5}{8} \sum_{i=1}^{n} \tau_i + \Ex{\bv}{\sum_{j \in \SOLD(\bv)} p_j}.
\end{align*}
Using Lemma~\ref{lem:lower-bound-on-core-utilities} to lower bound $\sum_{i=1}^{n} \Ex{\bv}{\mu_i(v_i, S_i(\bv))}$ and Lemma~\ref{lem:tau-and-postrev} to upper bound $\sum_{i=1}^{n} \tau_i$ we obtain
\begin{align*}
\textsc{APostEnRev} & \geq \frac{1}{4} \left( \frac{1}{\alpha} \textsc{Core} - \Ex{\bv} {\sum_{j \in
\SOLD(\bv)} p_j} - \frac{2 (1+b)}{b (1-b)} \cdot \textsc{PostRev} \right) \\
&\hspace*{120pt}- \frac{5}{8} \cdot \frac{4}{1-b} \cdot \textsc{PostRev} + \Ex{\bv}{\sum_{j \in \SOLD(\bv)} p_j} \\
& \geq \frac{1}{4 \alpha} \textsc{Core} - \frac{1 + 6b}{2b (1 - b)} \cdot \textsc{PostRev},
\end{align*}
as claimed.
\end{proof}

%% file: lower-bound.tex
% !TEX root = main.tex

%% LOWER BOUND
\section{Going Beyond $\mathbf{O(\text{log}\, \text{log}\, m})$}
\label{sec:lower-bound}
We leave it as an open problem whether the $O(\log \log m)$ factor could be further reduced, possibly even to a constant. As a matter of fact, many techniques presented in this paper seem to be very useful to reach such an improved guarantee. In this section we discuss to what extent they can be applied and where there are barriers.

In all of our proofs, the $O(\log \log m)$-factor originates from variants of the same technical lemma. In particular, the complete-information proof of Theorem~\ref{thm:upper-bound} as provided in Section~\ref{sec:complete-info} is mainly based on Lemma~\ref{lem:key-lemma}. Similar lemmas are used for all other proofs as well. This is why we will now revisit the proof of Lemma~\ref{lem:key-lemma}.

The proof of Lemma~\ref{lem:key-lemma} shows that its statement for any value of $\alpha$ is indeed equivalent to there being a vector of probabilities $\bq = (q_j)_{j \in U}$, $q_j \in [0,1]$, one for each item such that the $\lambda$ player has value at least $v_i(U)/\alpha$ in the zero-sum game induced by the $\bq$ vector. More formally, $g(\bq) \geq v_i(U)/\alpha$ for the function $g(\bq) = \max_{\lambda \in \Delta(\bq)} \min_{\mu \in \Delta(\bq)} \sum_{S,T \subseteq U} \lambda_S \mu_T v_i(S\setminus T)$, where $\Delta(\bq)$ denotes all probability distributions $\nu$ over sets $S \subseteq U$ such that $\sum_{T \ni j} \nu_T \leq q_j$ for all $j$. This also means that to show the existence of prices for an $o(\log \log m)$-approximation it suffices to show that there always is a vector $\bq$ such that $g(\bq) \geq v_i(U)/\alpha$ for $\alpha \in o(\log \log m)$.

Our proof continues by Lemma~\ref{lem:loglog-complete}, showing that there exists a $q \in [0,1]$ such that for $\bq = (q,\dots,q) \in [0,1]^{\lvert U \rvert}$ we have $g(\bq) \geq \frac{1}{\alpha} v_i(U)$ for $\alpha \in O(\log \log m)$. That is, we put the same probability mass $q$ on every item. As we will show now, the $O(\log \log m)$ bound is indeed tight for strategies that put the same probability mass $q$ on all items. In other words, for a $o(\log \log m)$ bound, one would have to devise a more sophisticated way to choose the $\bq$ vector.

%Old text: The critical step in our upper bound proof is the step in which we lower bound the value of the zero-sum game. In both the complete and incomplete information version of the proof, we lower bound the value of the zero-sum game by establishing the existence of a strategy $\lambda$ that puts the same probability mass $q$ on each item and showed that this way the $\lambda$-player can guarantee herself a $1/O(\log \log m)$ fraction of the optimal expected welfare.

\begin{theorem}
There exists a subadditive function $v_i$ such that for $U = M$ we have
\[
\max_{q \in [0,1]} g(q) = \max_{q \in [0,1]} \max_{\lambda \in \Delta(q,\dots,q)} \min_{\mu \in \Delta(q,\dots,q)} \sum_{S,T \subseteq M} \lambda_S \mu_T v_i(S\setminus T) \leq \frac{1}{\Omega(\log \log m)} v_i(M)
\]
\end{theorem}
\begin{proof}
\newcommand{\layerindex}{\ell}
\newcommand{\layernumbers}{L}
\newcommand{\copyindex}{b}
\newcommand{\copynumbers}{B}

As we are only interested in an asymptotic bound, we can assume without loss of generality that $m = 2^{2^\layernumbers}$ for some $\layernumbers \in \mathbb{N}$. That is, $\layernumbers = \log_2 \log_2 m$. The subadditive function $v_i$ will be the composed of $\layernumbers+1$ subadditive functions $v^{(\layerindex)}$ with $v_i(S) = \sum_{\layerindex=0}^{\layernumbers} v^{(\layerindex)}(S)$ for all $S \subseteq M$. The idea is that each of the functions $v^{(\layerindex)}$ requires a different value of $q$ so that the $\lambda$-player is guaranteed a good fraction of the value.

The subadditive function will be a judiciously chosen variant of the function that is used to show the $\Omega(\log m)$-separation between XOS and subadditive functions, which is based on the set cover integrality gap. The basic idea is to stack several such functions that operate on different subsets of the items on top of each other.

\paragraph{Construction of Valuation Function.} 

We construct the functions $v^{(\layerindex)}$ as follows. Recall that there are $m = 2^{2^\layernumbers}$ items in total. Let $\copynumbers^{(\layerindex)} = m/(2^{2^\layerindex}) = 2^{2^\layernumbers - 2^\layerindex}$. We partition all but $\copynumbers^{(\layerindex)}$ items into disjoint sets of equal size $M^{(\layerindex)}_1, \ldots, M^{(\layerindex)}_\copynumbers$. So, each of these sets $M^{(\layerindex)}_\copyindex$ has size $(m - \copynumbers^{(\layerindex)})/ \copynumbers^{(\layerindex)} = 2^{2^\layerindex} - 1$. On each of them, we use function $f_\copyindex\colon M^{(\layerindex)}_\copyindex \to \mathbb{R}_{\geq 0}$ defined as follows.

\begin{lemma}
\label{lemma:setcoverintegralitygapfunction}
Let $M'$ be a set of $2^k - 1$ items for some $k \in \mathbb{N}$. There is a subadditive function $f\colon 2^{M'} \to \mathbb{R}_{\geq 0}$ with the following properties:
\begin{enumerate}
\item $f(\{j\}) = 1$ for all $j \in M'$.
\item $f(M') \geq k$.
\item For every $d \in \{0, \ldots, k\}$, there is a family of subsets $\mathcal{D} \subseteq 2^{M'}$ such that  for all $D \in \mathcal{D}$ we have $\lvert D \rvert = 2^d - 1$ and $f(M' \setminus D) \leq k - d$. Furthermore, each $j \in M'$ is contained in the same number of sets $D \in \mathcal{D}$.
\end{enumerate}
\end{lemma}

\begin{proof}
This function was defined by Bhawalkar and Roughgarden \cite{BhawalkarR11}, who showed that it is only poorly approximated by XOS functions. It originates in the worst-case integrality gap for set cover linear programs (see, e.g, \cite[Example 13.4]{Vazirani2001}). We identify the set of items $M'$ with binary vectors $\mathbb{F}_2^k \setminus \{0\}$. Consider all binary vectors $\mathbf{i} \in \mathbb{F}_2^k \setminus \{0\}$. Let $S_i = \{j \mid \mathbf{j} \cdot \mathbf{i} = 1 \}$. For each set of items $T \subseteq M'$ let $f(T)$ be the minimum number of sets $S_i$ required to cover the items in $T \setminus \{ 0 \}$. 

The first property now follows from the following fact. Consider any $\mathbf{j} \in \mathbb{F}_2^k \setminus \{0\}$. As $\mathbf{j} \neq 0$, there has to be some $z$ for which $j_z = 1$. Now, choosing $\mathbf{i}$ to be the $z$-th unit vector, we have $\mathbf{j} \cdot \mathbf{i} = 1$. The second property is shown in \cite[Example~4.5]{BhawalkarR11}. For the third property, one chooses $\mathcal{D}$ to be the set of all $d$-dimensional subspaces of $\mathbb{F}_2^k$ excluding the all-zero vector. For this choice of $\mathcal{D}$, the property is verified in \cite[Lemma~4.3]{DuettingK17}.
\end{proof}

We then define $v^{(\layerindex)}(S) = \frac{1}{\copynumbers^{(\layerindex)}} \sum_{\copyindex = 1}^{\copynumbers^{(\layerindex)}} \frac{f_\copyindex(S \cap M^{(\layerindex)}_\copyindex)}{f_\copyindex(M^{(\layerindex)}_\copyindex)}$.

%\textcolor{red}{The subadditive function $v^{(\layerindex)}$ is the sum of $m/(2^{2^\layerindex})$ many set cover integrality gap functions that operate on $m/(2^{2^\layerindex})$ disjoint set of items of  $2^{2^\layerindex}$ items each.}
%
%\textcolor{orange}{We refer to these disjoint sets of items as \emph{copies}. We assign items to copies based on their binary vector representation. The first $k_i = \log_2(m/(2^{2^\layerindex})) = 2^\layernumbers - 2^\layerindex$ digits address the copy, and the remaining digits say which item in the copy it is:}
%
%\textcolor{orange}{\begin{align*}
%(\quad\underbrace{x_1, \dots, x_{k_i}}_{\text{which copy}}, \underbrace{x_{k_i+1}, \dots, x_{2^\layernumbers}}_{\text{which item in the copy}})
%\end{align*}}
%
%\textcolor{orange}{We note that the total value of $v^{(\layerindex)}$ defined this way would be $(m/(2^{2^\layerindex})) \cdot 2^\layerindex$, because there are $(m/(2^{2^\layerindex}))$ copies and the value that the set cover integrality gap function assigns to all $2^{2^\layerindex}$ items in a copy is $\log_2(2^{2^\layerindex}) = 2^\layerindex$.}
%
%\textcolor{orange}{It will be convenient to have $v^{(\layerindex)}(M) = 1$ for all $\layerindex$. So we scale down the valuations on level $\layerindex$ by multiplying them with a normalization factor of
%\[
%\frac{1}{(m/(2^{2^\layerindex})) \cdot 2^\layerindex} = \frac{2^{2^\layerindex}}{m \cdot 2^\layerindex}.
%\]}

This construction normalizes each function $v^{(\layerindex)}$ to $v^{(\layerindex)}(M) = 1$ and this way $v_i(M) = \sum_{\layerindex=0}^{\layernumbers} v^{(\layerindex)}(M) = \layernumbers+1$. Note that $v$ is a subadditive function because it is a weighted sum of subadditive functions.

It will be instructive to think of higher $\layerindex$ as ``more subadditive'' in the sense that the respective $v^{(\layerindex)}$ operate on a smaller number of sets $M^{(\layerindex)}_\copyindex$, each of which consists of more items. So there is less additivity because there are fewer sets $M^{(\layerindex)}_\copyindex$, and also less additivity because the respective functions $f^{(\layerindex)}_\copyindex$, which are ``very subadditive'', operate on larger sets of items.

\paragraph{Bound on Value.}

Consider some $q \in [0,1]$. Let $\layerindex^\star = \lceil \log_2 \log_2 \frac{1}{q} \rceil$. 

We will now construct the distribution $\mu$ as follows. Consider any $\layerindex > \layerindex^\star$ and $\copyindex \in [\copynumbers^{(\layerindex)}]$. Let $x = \layerindex - \layerindex^\star$. Let $D^{(\layerindex)}_\copyindex$ be a set drawn uniformly from the set $\mathcal{D}$ as it exists for function $f^{(\layerindex)}_\copyindex$ according to Lemma~\ref{lemma:setcoverintegralitygapfunction} for $d = \lfloor 2^{\layerindex^\star}(2^x-c^x) \rfloor$, where $c = \frac{3}{2}$. By this definition, for all $j \in M^{(\layerindex)}_\copyindex$, we have that $j \in D^{(\layerindex)}_\copyindex$ with probability
\[
	\Pr{j \in D^{(\layerindex)}_\copyindex} = \frac{\lvert D^{(\layerindex)}_\copyindex \rvert}{\lvert M^{(\layerindex)}_\copyindex \rvert} = \frac{2^d - 1}{2^{2^\layerindex} - 1} \leq \frac{2^d}{2^{2^\layerindex}}
\]
because every item is equally likely to be contained. Furthermore,
\[
\frac{2^d}{2^{2^\layerindex}} \leq \frac{2^{2^{\layerindex^\star}(2^x-c^x)}}{2^{2^\layerindex}} = \frac{2^{2^\layerindex-2^{\layerindex^\star}c^x}}{2^{2^\layerindex}} = \left(2^{-2^{\layerindex^\star}}\right)^{c^x} \leq q^{c^x} \leq q c^{x \ln q} \leq q \left(\frac{1}{2} \right)^x.
\]
Here, we use that $q^{z - e \ln z} \leq q$ and therefore $q^z \leq q z^{e \ln q}$ for all $z \geq 1$, which we apply for $z = c^x$; also, by, construction, $c^{e \ln q} \leq \frac{1}{2}$

Now, let $T$ be the union of all these random sets $D^{(\layerindex)}_\copyindex$ for $\layerindex > \layerindex^\star$. By union bound, each item is contained in $T$ with probability at most $q$. Therefore, the distribution of $T$ is a feasible choice for $\mu$.

Our goal is now to bound $\sum_{S,T \subseteq M} \lambda_S \mu_T v_i(S\setminus T)$ we split up this sum as follows
\[
\sum_{S,T \subseteq M} \lambda_S \mu_T v_i(S\setminus T) = \sum_{S,T \subseteq M} \lambda_S \mu_T \sum_{\layerindex = 0}^\layernumbers v^{(\layerindex)}(S\setminus T) \leq \sum_{\layerindex = 0}^{\layerindex^\star} \sum_S \lambda_S v^{(\layerindex)}(S) + \sum_{\layerindex = \layerindex^\star+1}^{\layernumbers} \sum_T \mu_T v^{(\layerindex)}(M\setminus T).
\]
Here, we use the bounds $v^{(\layerindex)}(S \setminus T) \leq v^{(\layerindex)}(S)$ and $v^{(\layerindex)}(S \setminus T) \leq v^{(\layerindex)}(M \setminus T)$.

For $\layerindex < \layerindex^\star$, we use subadditivity to get
\[
\sum_S \lambda_S v^{(\layerindex)}(S) \leq \sum_S \sum_{j \in S} \lambda_S v^{(\layerindex)}(\{j\}) = \sum_j v^{(\layerindex)}(\{j\}) \sum_{S:j \in S} \lambda_S \leq q \sum_j v^{(\layerindex)}(\{j\}).
\]
By the construction $v^{(\layerindex)}(\{j\}) = \frac{1}{(m/2^{2^\layerindex}) \cdot 2^\layerindex} \leq \frac{1}{2^\layerindex}$ if $j$ in included in one of the sets $M^{(\layerindex)}_1, \ldots, M^{(\layerindex)}_\copynumbers$, otherwise $v^{(\layerindex)}(\{j\}) = 0$. Consequently, as $2^{2^\layerindex} \leq \frac{1}{q}$ for $\layerindex < \layerindex^\star$,
\[
\sum_S \lambda_S v^{(\layerindex)}(S) \leq q \sum_j v^{(\layerindex)}(\{j\}) = q \cdot m \cdot \frac{1}{(m/2^{2^\layerindex}) \cdot 2^\layerindex} \leq \frac{1}{2^\layerindex}.
\]
For $\layerindex = \layerindex^\star$, we use that $v^{(\layerindex)}(S) \leq 1$ for all $\layerindex$. So the total value contributed by all $\layerindex \leq \layerindex^\star$ is 
\[
\sum_{\layerindex = 0}^{\layerindex^\star} \sum_S \lambda_S v^{(\layerindex)}(S) \leq 1 + \sum_{\layerindex=0}^{\infty} \frac{1}{2^\layerindex} = 3. 
\]

For $\layerindex > \layerindex^\star$ a different argument applies. Let $x = \layerindex - \layerindex^\star$. %Consider $\mu^{(\layerindex)}$ that removes subspaces of dimension $d = 2^{\layerindex^\star}(2^x-c^x)$ from each copy\textcolor{black}{, where $c = \frac{3}{2}$} If the subspace is drawn uniformly, then the probability mass on every item is
%\[
%	\frac{2^d}{2^{2^\layerindex}} = \frac{2^{2^{\layerindex^\star}(2^x-c^x)}}{2^{2^\layerindex}} = \frac{2^{2^\layerindex-2^{\layerindex^\star}c^x}}{2^{2^\layerindex}} = \left(2^{-2^{\layerindex^\star}}\right)^{c^x} = \textcolor{black}{q^{c^x} \leq q c^{x \ln q} \leq q \left(\frac{1}{2} \right)^x},
%\]
%
%Moreover, for any $T$ that consists of a subspace of dimension $d$ from each group,
Recall the definition $v^{(\layerindex)}(S) = \frac{1}{\copynumbers^{(\layerindex)}} \sum_{\copyindex = 1}^{\copynumbers^{(\layerindex)}} \frac{f_\copyindex(S \cap M^{(\layerindex)}_\copyindex)}{f_\copyindex(M^{(\layerindex)}_\copyindex)}$ and also, by definition, $(M \setminus T) \cap M^{(\layerindex)}_\copyindex \subseteq M^{(\layerindex)}_\copyindex \setminus D^{(\layerindex)}_\copyindex$. Thus
\[
v^{(\layerindex)}(M\setminus T) = \frac{1}{(m/2^{2^\layerindex})\cdot2^\layerindex} \cdot \frac{m}{2^{2^\layerindex}} \cdot (2^\layerindex-d) = \frac{1}{2^\layerindex} (2^\layerindex-d) = \frac{1}{2^\layerindex} (2^\layerindex- \lfloor 2^{\layerindex^\star} (2^x - c^x) \rfloor) \leq c^x 2^{\layerindex^\star} + \frac{1}{2^\layerindex} = \left(\frac{c}{2}\right)^x + \frac{1}{2^\layerindex} .
\]
This bound holds for all $T$ in the support of $\mu$ and so also $\sum_T \mu_T v^{(\layerindex)}(M\setminus T) \leq \left(\frac{c}{2}\right)^x$. Finally, taking the sum over all $\layerindex$ from $\layerindex^\star+1$ to $\layernumbers$, we have
\[
\sum_T \mu_T \sum_{\layerindex = \layerindex^\star+1}^{\layernumbers} v^{(\layerindex)}(M\setminus T) \leq \sum_{\layerindex = \layerindex^\star + 1}^{\layernumbers} \sum_T \mu_T^{(\layerindex)} v^{(\layerindex)}(M \setminus T) \leq \sum_{x=1}^{\infty} \left(\frac{c}{2}\right)^x + \sum_{\layerindex = 0}^\infty \frac{1}{2^\layerindex} = \frac{c}{2-c} + 2.
\]

Combining the bounds from the two cases shows that
\[
\sum_{S,T} \lambda_S \mu_S v_i(S\setminus T) = O(1).
\]
Combined with $v_i(M) = \layernumbers+1$ this shows the claim.
\end{proof}

%% file: appendix.tex
% !TEX root = main.tex

\section{Proof of Lemma~\ref{lem:key-lemma-bayesian}}
\label{app:proof-key-lemma-bayesian}

%As in the complete information proof let $\gamma = 1/\alpha$. 
To prove Lemma~\ref{lem:key-lemma-bayesian} we first follow the same steps as in the proof of Lemma~\ref{lem:key-lemma}. That is, we will capture the condition in a linear program and then use LP duality.
%While the first few steps that lead to the zero-sum game are very similar to the complete information case, the lower bound on the value of the zero-sum game requires several new ideas.

Let $\gamma = 1/\alpha$. To show Lemma~\ref{lem:key-lemma-bayesian}  we have to show that for every joint distribution $\D = \prod_{i=1}^{n} \D_i$ there are prices $p_j$ for $j \in M$ and distributions $\lambda^{i,\bv}$ over $S \subseteq M$ for all $i$ and $\bv$ such that for all $T \subseteq M$
\[
\sum_{j \in T} p_j + \Ex{\bv}{\sum_{i=1}^{n} \sum_S \lambda_S^{i,\bv} \left( v_i(S \setminus T) - \sum_{j \in S} p_j \right)} \geq \gamma \Ex{\bv}{\bv(\OPT(\bv))}
\]
or equivalently
\begin{align}
\Ex{\bv}{\sum_{i=1}^{n} \sum_S \lambda^{i,\bv}_S \sum_{j \in S} p_j} - \sum_{j \in T} p_j \leq \Ex{\bv}{\sum_{i=1}^{n} \sum_S \lambda^{i,\bv}_S v_i(S \setminus T)} - \gamma \Ex{\bv}{\bv(\OPT(\bv))}. \label{eq:sufficient-condition-bayesian}
\end{align}

To reformulate this condition through LP duality, fix distributions $\lambda^{i,\bv}$ for all $i$ and $\bv$, and consider the following LP with variables $p_j \geq 0$ for $j \in M$, $\ell_{+} \geq 0$, and $\ell_{-} \geq 0$ which maximizes the slack $\ell_{+}-\ell_{-}$ that we need to add to the left-hand side of \eqref{eq:sufficient-condition-bayesian} in order to satisfy the inequality
\begin{align*}
\text{max}\quad&\ell_{+}-\ell_{-}\\
\text{s.t.}\quad& \Ex{\bv}{\sum_{i=1}^{n} \sum_S \lambda^{i,\bv}_S \sum_{j \in S} p_j} - \sum_{j \in T} p_j +\ell_{+} - \ell_{-}\\
&\hspace*{50pt}\leq \Ex{\bv}{\sum_{i=1}^{n} \sum_S \lambda^{i,\bv}_S v_i(S \setminus T)} - \gamma \Ex{\bv}{\bv(\OPT(\bv))} &&\text{for all $T \subseteq M$}\\
&p_j \geq 0 &&\text{for all $j \in M$}\\
&\ell_{+} \geq 0\\
&\ell_{-} \geq 0. 
\end{align*}

The dual LP has variables $\mu_T \geq 0$ for every set $T \subseteq M$:
\begin{align*}
\text{min}\quad&\sum_{T} \mu_T \left(\Ex{\bv}{ \sum_{i=1}^{n} \sum_S \lambda^{i,\bv}_S v_i(S \setminus T)} - \gamma\Ex{\bv}{\bv(\OPT(\bv))}\right)\\
\text{s.t.}\quad&-\sum_{T: T \ni j} \mu_T + \sum_T \Ex{\bv}{\sum_{i=1}^{n} \sum_{S: S \ni j}\lambda^{i,\bv}_S} \mu_T \geq 0 &&\text{for all $j \in M$}\\
&\sum_{T} \mu_T = 1\\
&\mu_T \geq 0.
\end{align*}

We can interpret $\mu$ as a distribution over sets of items that puts at most the same probability mass on each item as the collection of distributions $\{\lambda^\bv\}_\bv$. Indeed, by combining the first with the second constraint we obtain that for all $j \in M$:
\begin{align*}
\sum_{T: T \ni j} \mu_T \leq \Ex{\bv}{ \sum_{i=1}^{n} \sum_{S: S \ni j} \lambda^{i,\bv}_S}.
\end{align*}

Now inequality (\ref{eq:sufficient-condition-bayesian}) is satisfied if the optimal solution to the primal LP is non-negative, which by strong duality is the case whenever all feasible solutions to the dual LP have non-negative value. This, in turn, is true whenever for all $\mu$ such that (i) $\sum_T \mu_T = 1$ and (ii) $\sum_{T: T \ni j} \mu_T \leq \Ex{\bv}{\sum_{i=1}^{n} \sum_{S: S \ni j} \lambda^{i,\bv}_S}$ for all $j$, we have
\[
\sum_{T} \mu_T \Ex{\bv}{\sum_{i=1}^{n} \sum_S \lambda^{i,\bv}_S v_i(S \setminus T)} \geq \gamma\Ex{\bv}{\bv(\OPT(\bv)}.
\]

%We will again view this as a zero-sum game where one player controls the $\lambda^{i,v_i}$ and the other the $\mu$. The $\lambda$-player seeks to maximize its payoff, and the $\mu$-player seeks to minimize it. We want to show that there is a strategy for the $\lambda$-player that guarantees her a $\gamma$ fraction of the expected optima welfare $\gamma \Ex{\bv}{\bv(\OPT(\bv))}$. 

Now, as in the case of complete information, we'd like to establish this inequality by interpreting it as a zero-sum game. However, the arguments in Lemma~\ref{lem:loglog-complete} for lower-bounding the value of the zero sum game require the game to be played for a fixed valuation profile. We therefore add an intermediate step in which we lower bound the value of the zero sum game we are interested in by the expected value of related zero sum games for fixed valuation profiles.

%As in the complete information case we will establish this by considering an appropriate zero-sum game, except that now the argument is a bit more complicated as the arguments in Lemma~\ref{lem:loglog-complete} require the zero-sum game to be played for a fixed valuation profile. We will therefore lower bound the value of the zero-sum game we are interested in by the expected value of related zero-sum games for fixed valuation profiles. 

%\textcolor{teal}{Paul: I think it wouldn't hurt to go over the following once more to see if everything is defined appropriately.}

The zero-sum game we are interested in is the following. Write $\lambda^{\bv}$ for $\{\lambda^{i,\bv}\}_{i=1}^{n}$. For $q \in [0,1]$ let $\Lambda(q) = \{ \{\nu^\bv\}_{\bv} \mid \Ex{\bv}{\sum_{i=1}^{n} \sum_{S: S \ni j} \nu^{i,\bv}_S} \leq q\}$. Let $\Delta(q,\dots,q)$ be defined as in the complete information case. We want to show
%\begin{align}
\[
	g(q) = \sup_{\{\lambda^\bv\}_\bv \in \Lambda(q)} \inf_{\mu \in \Delta(q, \dots, q)} \sum_{T} \mu_T \Ex{\bv}{\sum_{i=1}^{n} \sum_S \lambda^{i,\bv}_S v_i(S \setminus T)} \geq \gamma \Ex{\bv}{\bv(\OPT(\bv))}. %\label{eq:original-zero-sum}
\]
%\end{align}

To show this consider the following class of zero-sum games, which have the same basic structure, but instead of considering the payoff in expectation over valuation profiles the payoffs are now for a given and fixed valuation profile $\bv$.
For $q \in [0,1]$ let $\Gamma(q) = \{\{\nu^{i}\}_{i=1}^{n} \mid \sum_{i=1}^{n} \sum_{S: S \ni j} \nu^{i}_S \leq q \;\text{for all $j$ in $M$}\}$. Define
\[
g^{\bv}(q) = \sup_{\lambda \in \Gamma(q)} \inf_{\mu \in \Delta(q,\dots,q)} \sum_{T} \mu_T \left(\sum_{i=1}^{n} {\sum_S \lambda^{i}_S v_i(S \setminus T)} \right).
\]

Now consider the definition of $g(q)$. If we swap the infimum and the expectation we only give the $\mu$-player more power. Hence for any $q$, %\textcolor{blue}{TODO: explain} %%$g(q)$ is lower bounded by $\Ex{\bv}{g_{\bv}(q)}$. 
%%
%The basic idea now will be to show that there is a $q$ for which the expected value $\Ex{\bv}{g_{\bv}(q)}$ of the zero-sum games for fixed valuation profiles is lower bounded by $\gamma \Ex{\bv}{\bv{\OPT(\bv)}}$, and exploit that $g(q)$ is lower bounded by $\Ex{\bv}{g_{\bv}(q)}$:
\begin{align*}
	g(q) 
	&= \sup_{\{\lambda^\bv\}_\bv \in \Lambda(q)} \inf_{\mu \in \Delta(q, \dots, q)} \sum_{T} \mu_T \Ex{\bv}{\sum_{i=1}^{n} \sum_S \lambda^{i,\bv}_S v_i(S \setminus T)}\\
	&\ge \sup_{\{\lambda^\bv\}_\bv \in \Lambda(q)}  \Ex{\bv}{\inf_{\mu \in \Delta(q, \dots, q)} \sum_{T} \mu_T\sum_{i=1}^{n} \sum_S \lambda^{i,\bv}_S v_i(S \setminus T)}\\
	&= \Ex{\bv}{\sup_{\lambda \in \Gamma(q)} \inf_{\mu \in \Delta(q,\dots,q)} \sum_{T} \mu_T \left(\sum_{i=1}^{n} {\sum_S \lambda^{i}_S v_i(S \setminus T)} \right)}\\
	&=\Ex{\bv}{g^{\bv}(q)}%,
\end{align*}
%where we first used linearity of expectation and Jensen's inequality to push the expectation over the infimum, and then that by swapping the supremum with the expectation we only give the $\mu$-player more power.

So in order to establish the existence of a $q$ for which $g(q) \geq \gamma \Ex{\bv}{\bv(\OPT(\bv))}$ it suffices to show that there is a $q$ for which $\Ex{\bv}{g^{\bv}(q)} \geq \gamma \Ex{\bv}{\bv(\OPT(\bv))}$.

\begin{lemma}\label{lem:loglog-incomplete}
There exists a $q \in [0,1]$ such that
\[
\Ex{\bv}{g^{\bv}(q)} \geq \frac{1}{\alpha}\bv(\OPT(\bv)).
\]
\end{lemma}

In order to prove Lemma~\ref{lem:loglog-incomplete} we will basically proceed as in the proof of Lemma~\ref{lem:loglog-complete}. However, in order to show that there is a single $q$ that works for $\Ex{\bv}{g^{\bv}(q)}$ it will be useful to separate out the following auxiliary lemma.

\begin{lemma}\label{lem:average}
For every $\bv$ and $\ell = \log \log m$
\[
	\frac{1}{\ell+1}\sum_{i=0}^{\ell} g^{\bv}\left(2^{-2^i}\right) \geq \frac{1}{\ell+1} \cdot \left(\frac{1}{2}-\frac{1}{m}\right) \cdot \bv(\OPT(\bv)).
\]
\end{lemma}

\begin{proof}[Proof of Lemma~\ref{lem:average}]
As in the proof of Lemma~\ref{lem:loglog-complete} our strategy will be to derive a lower bound on $g^{\bv}(q)$ in terms of an appropriately chosen $f$-function. The $f$-function we will use is %the following. %For $q \in [0,1]$ define
%\begin{align*}
%	f^{\bv}(q) &= \sup_{\lambda\in \Gamma(q)} \sum_{i=1}^{n} \sum_{S} \lambda^{i,\bv}_S v_i(S).
%\end{align*}
\begin{align*}
	f^{\bv}(q) &= \sup_{\lambda\in \Gamma(q)} \sum_{i=1}^{n} \sum_{S} \lambda^{i}_S v_i(S).
\end{align*}

We first use that for any of the involved valuation functions $v_i$ and any two sets $S$ and $T$ subadditivity implies that $v_i(S \setminus T) \geq v_i(S) - v_i(S \cap T)$ to derive the following lower bound on $g^{\bv}(q)$:
\begin{align*}
	g^{\bv}(q) 
	&= \sup_{\lambda \in \Gamma(q)} \inf_{\mu \in \Delta((q,\dots,q))} \sum_{T} \mu_T \left(\sum_{i=1}^{n} \sum_S \lambda^{i}_S v_i(S \setminus T) \right)\\
	&\geq \sup_{\lambda \in \Gamma(q)} \inf_{\mu \in \Delta((q,\dots,q))} \sum_{T} \mu_T \left(\sum_{i=1}^{n} \sum_S \lambda^{i}_S \Big(v_i(S) - v_i(S \cap T)\Big) \right).
\end{align*}
We can now regroup the right-hand side by splitting into two summations (one over $v_i(S)$ and the other over $v_i(S \cap T)$), use that $\sum_T \mu_T = 1$ to simplify, and then pull the infimum over the first of the two summations and the minus sign:
\begin{align*}
	g^{\bv}(q) &\geq \sup_{\lambda \in \Gamma(q)} \inf_{\mu \in \Delta((q,\dots,q))} \left( \sum_{T} \mu_T \left(\sum_{i=1}^{n} \sum_S \lambda^{i}_S v_i(S) \right) - \sum_{T} \mu_T \left(\sum_{i=1}^{n} \sum_S \lambda^{i}_S v_i(S \cap T) \right) \right)\\
	&= \sup_{\lambda \in \Gamma(q)} \inf_{\mu \in \Delta((q,\dots,q))} \left( \sum_{i=1}^{n} \sum_S \lambda^{i}_S v_i(S) - \sum_{T} \mu_T \left(\sum_{i=1}^{n} \sum_S \lambda^{i}_S v_i(S \cap T) \right) \right)\\
	&= \sup_{\lambda \in \Gamma(q)} \left( \sum_{i=1}^{n} \sum_S \lambda^{i}_S v_i(S) - \sup_{\mu \in \Delta((q,\dots,q))} \sum_{T} \mu_T \left(\sum_{i=1}^{n} \sum_S \lambda^{i}_S v_i(S \cap T) \right) \right).
\end{align*}

We next interpret $\lambda^{i}_S \cdot \mu_T$ as a probability distribution over sets $S \cap T$. The probability mass that this distribution puts on each item is at most $q^2$. We can thus lower bound the right-hand side in the previous equation and hence $g^{\bv}(q)$ as follows:
\begin{align*}
	g^{\bv}(q) 
	&\geq \sup_{\lambda \in \Gamma(q)} \left( \sum_{i=1}^{n} \sum_S \lambda^{i}_S v_i(S) - \sup_{\nu \in \Gamma(q^2)} \left(\sum_{i=1}^{n} \sum_S \nu^{i}_S v_i(S) \right) \right)
	= f^{\bv}(q) - f^{\bv}(q^2).
\end{align*}

For any $\ell$ we thus have,
\begin{align*}
\sum_{i=0}^{\ell} g^{\bv}\left(2^{-2^i}\right) \geq \sum_{i=0}^{\ell} \Bigg( f^{\bv}\left(2^{-2^i}\right) - f^{\bv}\left(2^{-2^{i+1}}\right)\Bigg) = f^{\bv}\left(2^{-1}\right) - f^{\bv}\left(2^{-2^{\ell+1}}\right),
\end{align*}
by a telescoping sum argument.

Moreover, for $\ell = \log \log m$,
\begin{align*}
	2^{-2^{\ell+1}} = 2^{-2 \log m} = \frac{1}{m^2}.
\end{align*}
%With $\ell = \log \log |U|$,
%\begin{align*}
%	2^{-2^{\ell+1}} = 2^{-2 \log |U|} = \frac{1}{|U|^2}.
%\end{align*}

So for $\ell = \log \log m$,
\begin{align*}
\sum_{i=0}^{\ell} g^{\bv}\left(2^{-2^i}\right) \geq f^{\bv}\left(\frac{1}{2}\right) - f^{\bv}\left(\frac{1}{m^2}\right).
\end{align*}

We next derive a lower bound on $f^{\bv}(1/2)$ and an upper bound of $f^{\bv}(1/m^2)$. For the lower bound on $f^{\bv}(1/2)$ observe that one possible choice for $\lambda^{i}$ for $i = 1, \dots, n$ would be to choose $\OPT_i(\bv)$ with probability $1/2$. Hence
\begin{align*}
	f^{\bv}\left(\frac{1}{2}\right) &\geq \frac{1}{2} \cdot \bv(\OPT(\bv)).
\end{align*}

On the other hand, we can derive an upper bound on $f^{\bv}(1/m^2)$ from the (very crude) upper bound $v_i(S) \leq \bv(\OPT(\bv))$ for all $i$. Namely, 
\begin{align*}
	f^{\bv}\left(\frac{1}{m^2}\right) 
	&= \sup_{\lambda \in \Gamma\left(\frac{1}{m^2}\right)} \sum_{i=1}^{n} \sum_{S} \lambda^{i}_S v_i(S)\\
	&\leq \sup_{\lambda \in \Gamma\left(\frac{1}{m^2}\right)} \left(\sum_{i=1}^{n} \sum_{S \neq \emptyset} \lambda^{i}_S\right) \cdot \bv(\OPT(\bv))\\
	&\leq m \cdot \frac{1}{m^2} \cdot \bv(\OPT(\bv))\\
	&= \frac{1}{m} \cdot \bv(\OPT(\bv)).
\end{align*}

We conclude that
\begin{align*}
\frac{1}{\ell+1} \sum_{i=0}^{\ell} g(2^{-2^i}) \geq \frac{1}{\ell+1} \cdot \left(\frac{1}{2} - \frac{1}{m}\right) \cdot \bv(\OPT(\bv)),
\end{align*}
as claimed.
\end{proof}

It remains to show Lemma~\ref{lem:loglog-incomplete}.

\begin{proof}[Proof of Lemma~\ref{lem:loglog-incomplete}]
Let $\ell = \log \log m$. Taking expectations over $\bv$ on both sides of Lemma~\ref{lem:average} shows that for $j$ drawn uniformly from $\{0\} \cup [\ell]$ and $q_j = 2^{-2^j}$,
%averaging over $q_i = 2^{-2^i}$ for $i = 0, \dots, \ell$, 
\begin{align*}
\Ex{\bv,j}{g^{\bv}\left(q_j\right)} 
&\geq \frac{1}{\ell+1} \cdot \left(\frac{1}{2}-\frac{1}{m}\right) \cdot \Ex{\bv}{\bv(\OPT(\bv))}.%\\
%&= \frac{1}{O(\log \log m)} \cdot \Ex{\bv}{\bv(\OPT(\bv))}.
\end{align*}
Because this inequality holds in expectation over $j$ drawn uniformly from $\{0\} \cup [\ell]$ there must be a $j$ for which
\begin{align*}
\Ex{\bv}{g^{\bv}\left(q_j\right)} 
&\geq \frac{1}{\ell+1} \cdot \left(\frac{1}{2}-\frac{1}{m}\right) \cdot \Ex{\bv}{\bv(\OPT(\bv))},%\\
\end{align*}
as claimed.
\end{proof}

\section{Proof of Theorem~\ref{thm:upper-bound} (incomplete information)}
\label{app:proof-thm-upper-bound}

%\begin{proof}[Proof of Theorem~\ref{thm:upper-bound} (incomplete information)]
Let $\OPT(\bv)$ denote the welfare-maximizing allocation
given valuations $\bv$ and let $\ALG(\bv)$ denote the allocation of the
posted-price mechanism that uses the prices $p_j$ for $j \in M$ whose
existence is established in Lemma~\ref{lem:key-lemma-bayesian}.

%Specifically, throughout this theorem let $\bv$ and $\bv'$ denote two
%independent draws from $\D$.

For the lower bound on the utilities consider an arbitrary buyer $i$.
For any valuation profiles $\bv$ and $\bv'$, buyer $i$ with valuation $v_i$ can
consider drawing a set $S$ from $\lambda^{i,(v_i, \bv'_{-i})}$ and buy
whatever is left from $S$. Clearly, for any $v_i$, the set of items sold
on valuation profile $\bv$ before buyer $i$ is a subset of
$\SOLD(v'_i,\bv_{-i})$ which is what would be sold on valuation profile
$(v'_i,\bv_{-i})$. Since this holds for any $\bv'$ it also holds in
expectation when $\bv'$ drawn from $\D$, fixing $\bv$. %Let $T = \SOLD(v'_i,\bv_{-i})$.
So, since $v'_i$ and $\bv_{-i}'$ are independent,
\begin{align*}
 u_i((v_i,\bv_{-i}),\bp)
 &\geq
 \Ex{\bv'}{\sum_{S} \lambda^{i,(v_i, \bv'_{-i})}_S \left(v_i(S
\setminus \SOLD(v'_i,\bv_{-i})) - \sum_{j \in S} p_j\right)}.
\end{align*}

We can now take expectations over $\bv$ on both sides and exploit that
$\bv_{-i}$ and $\bv'_{-i}$ are identically and independently distributed to obtain
%, and because the distributions $\lambda^{i,v_i}$ that appear on the
%right-hand side only depend on $i$ and $v_i$ we can replace the
%$\bv_{-i}$ that appears on the right-hand side with $\bv'_{-i}$. We obtain,
\begin{align*}
 %\sum_{i=1}^{n} %\Ex{v_i \sim \D_i} {
 \Ex{\bv}{u_i((v_i,\bv_{-i}),\bp)} %}
 &\geq %\sum_{i=1}^{n}
 %\Ex{v_i \sim \D_i}{
 \Ex{\bv,\bv'}{\sum_{S} \lambda^{i,(v_i, \bv'_{-i})}_S \left(v_i(S
\setminus \SOLD(v'_i,\bv_{-i})) - \sum_{j \in S} p_j\right)}\\ %}
 &= \Ex{\bv,\bv'}{\sum_{S} \lambda^{i,(v_i, \bv_{-i})}_S \left(v_i(S
\setminus \SOLD(v'_i,\bv'_{-i})) - \sum_{j \in S} p_j\right)}.
\end{align*}

Summing this inequality over all buyers $i$ gives
\[
\sum_{i=1}^{n} \Ex{\bv}{u_i((v_i,\bv_{-i}),\bp)}
\geq
\sum_{i=1}^{n}
\Ex{\bv,\bv'}{\sum_{S} \lambda^{i,\bv}_S \left(v_i(S \setminus
\SOLD(\bv')) - \sum_{j \in S} p_j\right)}.
\]

Since Lemma~\ref{lem:key-lemma-bayesian} holds pointwise for any $T$
that is consistent across buyers, it also applies in expectation if we
draw $\bv'$ from $\D$ and set $T = \SOLD(\bv')$. Combining this with the
previous inequality we obtain
\begin{align}
\label{eq:utility-bayesian}
\sum_{i=1}^{n} \Ex{\bv}{u_i((v_i,\bv_{-i}),\bp)}
& \geq
\frac{1}{\alpha}\Ex{\bv}{\OPT(\bv)} - \Ex{\bv'} {\sum_{j \in
\SOLD(\bv')} p_j}. %\\
%& = \frac{1}{\alpha}\Ex{\bv}{\OPT(\bv)} - \Ex{\bv} {\sum_{j \in
%\SOLD(\bv)} p_j}.
%\label{eq:utility-bayesian}
\end{align}

For the revenue we have
\begin{align}
    \Ex{\bv \sim \D}{r(\bv,\bp)} = \Ex{\bv\sim \D}{\sum_{j \in
\SOLD(\bv)} p_j}. \label{eq:revenue-bayesian}
\end{align}

Adding (\ref{eq:revenue-bayesian}) to (\ref{eq:utility-bayesian}) shows
the claim.
%\end{proof}

\section{Proof of Theorem~\ref{thm:upper-bound-polytime}}
\label{app:proof-thm-upper-bound-polytime}

%\begin{proof}[Proof of Theorem~\ref{thm:upper-bound-polytime}]
We use $\OPT(\bv)$ to denote the welfare-maximizing allocation
given valuations $\bv$ and $\ALG(\bv)$ to denote the allocation of the
posted-price mechanism that uses the prices $p_j$ for $j \in M$ whose
existence is established in Lemma~\ref{lem:key-lemma-bayesian-polytime}.

%Specifically, throughout this theorem let $\bv$ and $\bv'$ denote two
%independent draws from $\D$.

For the lower bound on the utilities consider an arbitrary buyer $i$.
For any valuation profiles $\bv$ and $\bv'$, buyer $i$ with valuation $v_i$ can
consider drawing a set $S$ from $\lambda^{i,(v_i, \bv'_{-i})}$ and buy
whatever is left from $S$. Clearly, for any $v_i$, the set of items sold
on valuation profile $\bv$ before buyer $i$ is a subset of
$\SOLD(v'_i,\bv_{-i})$ which is what would be sold on valuation profile
$(v'_i,\bv_{-i})$. Since this holds for any $\bv'$ it also holds in
expectation when $\bv'$ drawn from $\D$, fixing $\bv$. %Let $T = \SOLD(v'_i,\bv_{-i})$.
So, since $v'_i$ and $\bv_{-i}'$ are independent,
\begin{align*}
 u_i((v_i,\bv_{-i}),\bp)
 &\geq
 \Ex{\bv'}{\sum_{S} \lambda^{i,(v_i, \bv'_{-i})}_S \left(v_i(S
\setminus \SOLD(v'_i,\bv_{-i})) - \sum_{j \in S} p_j\right)}.
\end{align*}

We can now take expectations over $\bv$ on both sides and exploit that
$\bv_{-i}$ and $\bv'_{-i}$ are identically and independently distributed to obtain
%, and because the distributions $\lambda^{i,v_i}$ that appear on the
%right-hand side only depend on $i$ and $v_i$ we can replace the
%$\bv_{-i}$ that appears on the right-hand side with $\bv'_{-i}$. We obtain,
\begin{align*}
 %\sum_{i=1}^{n} %\Ex{v_i \sim \D_i} {
 \Ex{\bv}{u_i((v_i,\bv_{-i}),\bp)} %}
 &\geq %\sum_{i=1}^{n}
 %\Ex{v_i \sim \D_i}{
 \Ex{\bv,\bv'}{\sum_{S} \lambda^{i,(v_i, \bv'_{-i})}_S \left(v_i(S
\setminus \SOLD(v'_i,\bv_{-i})) - \sum_{j \in S} p_j\right)}\\ %}
 &= \Ex{\bv,\bv'}{\sum_{S} \lambda^{i,(v_i, \bv_{-i})}_S \left(v_i(S
\setminus \SOLD(v'_i,\bv'_{-i})) - \sum_{j \in S} p_j\right)}.
\end{align*}

Summing this inequality over all buyers $i$ gives
\[
\sum_{i=1}^{n} \Ex{\bv}{u_i((v_i,\bv_{-i}),\bp)}
\geq
\sum_{i=1}^{n}
\Ex{\bv,\bv'}{\sum_{S} \lambda^{i,\bv}_S \left(v_i(S \setminus
\SOLD(\bv')) - \sum_{j \in S} p_j\right)}.
\]

Since Lemma~\ref{lem:key-lemma-bayesian-polytime} holds pointwise for any $T$
that is consistent across buyers, it also applies in expectation if we
draw $\bv'$ from $\D$ and set $T = \SOLD(\bv')$. Combining this with the
previous inequality we obtain
\begin{align}
\label{eq:utility-bayesian-polytime}
\sum_{i=1}^{n} \Ex{\bv}{u_i((v_i,\bv_{-i}),\bp)}
& \geq
\frac{1}{\alpha}\Ex{\bv}{\OPT(\bv)} - \Ex{\bv'} {\sum_{j \in
\SOLD(\bv')} p_j} - \epsilon. %\\
%& = \frac{1}{\alpha}\Ex{\bv}{\OPT(\bv)} - \Ex{\bv} {\sum_{j \in
%\SOLD(\bv)} p_j}.
%\label{eq:utility-bayesian}
\end{align}

For the revenue we have
\begin{align}
    \Ex{\bv \sim \D}{r(\bv,\bp)} = \Ex{\bv\sim \D}{\sum_{j \in
\SOLD(\bv)} p_j}. \label{eq:revenue-bayesian-polytime}
\end{align}

Adding (\ref{eq:revenue-bayesian-polytime}) to (\ref{eq:utility-bayesian-polytime}) shows
the claim.
%\end{proof}

\section{Proof of Lemma~\ref{lem:key-lemma-revenue}}
\label{app:proof-key-lemma-revenue}

%\begin{proof}[Proof of Lemma~\ref{lem:key-lemma-revenue}] 

%\textcolor{teal}{Paul: Changed $v_i$ to $v'_i$ in this proof, please check}
The following is a sketch.  The proof will follow the proof of Lemma~\ref{lem:key-lemma-bayesian} with minor changes.  For convenience, write $\OPT^{\mathbf{z}} = \Ex{\bv}{ \max_{\overline{\lambda} \in \Lambda(\mathbf{z})} \sum_{i,S} \overline{\lambda}^{i,\bv}_S v'_i(S)}$

First, when considering the LP with variables $p_j$, $\ell_+$, and $\ell_-$, we replace instances of $\Ex{\bv}{\bv(\OPT(\bv))}$ with $\OPT^{\mathbf{z}}$, and similarly for the dual LP.  Then just as in Lemma~\ref{lem:key-lemma-bayesian}, it suffices to show that there is some choice of $\lambda \in \Lambda(\mathbf{z})$ such that, for all $\mu$ such that (i) $\sum_T \mu_T = 1$ and (ii) $\sum_{T \ni j} \mu_T \leq \Ex{\bv}{\sum_i \sum_{S \ni j} \lambda_S^{i,\bv}}$, we have
\[ 
	\sum_T \mu_T \Ex{\bv}{\sum_{i,S} \lambda_S^{i,\bv} v'_i(S \backslash T)} \geq \gamma \OPT^{\mathbf{z}}. 
\]
We establish this using a zero-sum game.  Write $z_j(v_i) = \sum_i z_{ij}(v_i)$.  Fix $\bv$, and for $q \in [0,1]$, let $\Gamma(q) = \{\nu \ |\ \sum_i \sum_{S \ni j}\nu_S^i \leq q \cdot z_j(v_i) \text{ for all $j \in M$ }\}$.  Under this definition of $\Gamma(q)$, we define $g^{\bv}(q)$ in the same way as in Lemma~\ref{lem:key-lemma-bayesian} except that we replace $v_i$ with $v'_i$. I.e.,
\[
g^{\bv}(q) = \sup_{\lambda \in \Gamma(q)} \inf_{\mu \in \Delta((q,\dots,q))} \sum_{T} \mu_T \left(\sum_{i=1}^{n} {\sum_S \lambda^{i}_S v'_i(S \setminus T)} \right).
\]

We then define $\Lambda(q)$ in the same way as $\Lambda(\mathbf{z})$, but with the right-hand side of all inequalities multiplied by $q$. I.e.,
\[ 
	\Lambda(q) = \left\{ \{\lambda^\bv\}_\bv \colon \Ex{\bv_{-i}}{\sum_{S \ni j} \lambda^{i,\bv}_{S}} \leq q z_{ij}(v_i) \ \forall\ i,j, v_i \right\} 
\]
We use this definition of $\Lambda(q)$ to define $g(q)$, again replacing $v_i$ with $v'_i$:
\[
	g(q) = \sup_{\{\lambda^\bv\}_\bv \in \Lambda(q)} \inf_{\mu \in \Delta(q, \dots, q)} \sum_{T} \mu_T \Ex{\bv}{\sum_{i=1}^{n} \sum_S \lambda^{i,\bv}_S v'_i(S \setminus T)}.
\]	
Then precisely as in Lemma~\ref{lem:key-lemma-bayesian}, we have that $g(q) \geq \Ex{\bv}{g^{\bv}(q)}$.  
%\textcolor{teal}{Brendan: check that there are no strange implications by taking expectations over only $\bv_{-i}$ in the definition of $\Lambda(q)$}

The last step is to show that there is a choice of $q$ for which $\Ex{\bv}{g^{\bv}(q)} \geq \gamma \OPT^{\mathbf{z}}$.  This follows by a telescoping argument as before. The key observation is that if $\lambda \in \Gamma(q)$ and $\mu \in \Delta(q)$, then the distribution over $S \cap T$, where $S \sim \lambda$ and $T \sim \mu$, is in $\Gamma(q^2)$.  %\textcolor{teal}{Paul: I changed $\Lambda$ to $\Gamma$, please double check}

Defining $f^{\bv}$ with respect to $\Gamma(q)$ and $\bv'$, i.e,
\begin{align*}
	f^{\bv}(q) &= \sup_{\lambda\in \Gamma(q)} \sum_{i=1}^{n} \sum_{S} \lambda^{i}_S v'_i(S).
\end{align*}
we end up with a bound in terms of $(f^{\bv}(1/2) - f^{\bv}(1/m^2))$.  We have that $f^{\bv}(1/2) \geq (1/2)\OPT^{\mathbf{z}}$ in the same way as before.  We also have $f^{\bv}(1/m^2) \leq (1/m) \OPT^{\mathbf{z}}$, using the crude upper bound that each individual agent's contribution to the welfare under $\Lambda(\mathbf{z})$ is at most $\OPT^{\mathbf{z}}$.  Our analysis of the telescoping sum therefore follows in the same way as Lemma~\ref{lem:key-lemma-bayesian}, and the result follows.
%\end{proof}

%% file: main.bbl
\begin{thebibliography}{10}

\bibitem{AssadiS19}
Sepehr Assadi and Sahil Singla.
\newblock Improved truthful mechanisms for combinatorial auctions with
  submodular bidders.
\newblock In {\em Proccedings of the 60th {IEEE} Symposium on Foundations of
  Computer Science}, pages 233--248, 2019.

\bibitem{BILW14}
Moshe Babaioff, Nicole Immorlica, Brendan Lucier, and S.~Matthew Weinberg.
\newblock A simple and approximately optimal mechanism for an additive buyer.
\newblock In {\em Proceedings of the 55th {IEEE} Symposium on Foundations of
  Computer Science}, pages 21--30, 2014.

\bibitem{BhawalkarR11}
Kshipra Bhawalkar and Tim Roughgarden.
\newblock Welfare guarantees for combinatorial auctions with item bidding.
\newblock In {\em Proceedings of the 22nd {ACM-SIAM} Symposium on Discrete
  Algorithms}, pages 700--709, 2011.

\bibitem{CaiDW16}
Yang Cai, Nikhil~R. Devanur, and S.~Matthew Weinberg.
\newblock A duality based unified approach to bayesian mechanism design.
\newblock In {\em Proceedings of the 48th {ACM} Symposium on Theory of
  Computing}, pages 926--939, 2016.

\bibitem{CaiZ17}
Yang Cai and Mingfei Zhao.
\newblock Simple mechanisms for subadditive buyers via duality.
\newblock In {\em Proceedings of the 49th {ACM} Symposium on Theory of
  Computing}, pages 170--183, 2017.

\bibitem{chawla2007algorithmic}
Shuchi Chawla, Jason~D Hartline, and Robert Kleinberg.
\newblock Algorithmic pricing via virtual valuations.
\newblock In {\em Proceedings of the 8th ACM Conference on Electronic
  Commerce}, pages 243--251, 2007.

\bibitem{ChawlaHMS10}
Shuchi Chawla, Jason~D. Hartline, David~L. Malec, and Balasubramanian Sivan.
\newblock Multi-parameter mechanism design and sequential posted pricing.
\newblock In {\em Proceedings of the 42nd {ACM} Symposium on Theory of
  Computing}, pages 311--320, 2010.

\bibitem{chawla2010power}
Shuchi Chawla, David~L Malec, and Balasubramanian Sivan.
\newblock The power of randomness in bayesian optimal mechanism design.
\newblock In {\em Proceedings of the 11th ACM Conference on Electronic
  Commerce}, pages 149--158, 2010.

\bibitem{chawla2016mechanism}
Shuchi Chawla and J~Benjamin Miller.
\newblock Mechanism design for subadditive agents via an ex ante relaxation.
\newblock In {\em Proceedings of the 17th ACM Conference on Economics and
  Computation}, pages 579--596, 2016.

\bibitem{ChawlaMT19}
Shuchi Chawla, J.~Benjamin Miller, and Yifeng Teng.
\newblock Pricing for online resource allocation: Intervals and paths.
\newblock In Timothy~M. Chan, editor, {\em Proceedings of the 30th {ACM-SIAM}
  Symposium on Discrete Algorithms}, pages 1962--1981, 2019.

\bibitem{christodoulou2016bayesian}
George Christodoulou, Annam{\'a}ria Kov{\'a}cs, and Michael Schapira.
\newblock Bayesian combinatorial auctions.
\newblock {\em Journal of the ACM}, 63(2):1--19, 2016.

\bibitem{Dobzinski07}
Shahar Dobzinski.
\newblock Two randomized mechanisms for combinatorial auctions.
\newblock In {\em Proceedings of the 10th/11th APPROX-RANDOM Workshop}, pages
  89--103, 2007.

\bibitem{Dobzinski16}
Shahar Dobzinski.
\newblock Breaking the logarithmic barrier for truthful combinatorial auctions
  with submodular bidders.
\newblock In {\em Proceedings of the 48th ACM Symposium on Theory of
  Computing}, pages 940--948, 2016.

\bibitem{DuettingFKL17}
Paul D\"utting, Michal Feldman, Thomas Kesselheim, and Brendan Lucier.
\newblock Prophet inequalities made easy: {S}tochastic optimization by pricing
  non-stochastic inputs.
\newblock In {\em Proceedings of the 58th IEEE Symposium on Foundations of
  Computer Science}, pages 540--551, 2017.

\bibitem{DuettingK17}
Paul D{\"{u}}tting and Thomas Kesselheim.
\newblock Best-response dynamics in combinatorial auctions with item bidding.
\newblock In {\em Proceedings of the 28th {ACM-SIAM} Symposium on Discrete
  Algorithms}, pages 521--533, 2017.

\bibitem{DuttingK15}
Paul D{\"{u}}tting and Robert Kleinberg.
\newblock Polymatroid prophet inequalities.
\newblock In {\em Proceedings of the 23rd Annual European Symposium on
  Algorithms}, pages 437--449, 2015.

\bibitem{EhsaniHKS18}
Soheil Ehsani, MohammadTaghi Hajiaghayi, Thomas Kesselheim, and Sahil Singla.
\newblock Prophet secretary for combinatorial auctions and matroids.
\newblock In {\em Proceedings of the 29th {ACM-SIAM} Symposium on Discrete
  Algorithms}, pages 700--714, 2018.

\bibitem{Feige09}
Uriel Feige.
\newblock On maximizing welfare when utility functions are subadditive.
\newblock {\em {SIAM} Journal on Computing}, 39(1):122--142, 2009.

\bibitem{FeldmanFGL13}
Michal Feldman, Hu~Fu, Nick Gravin, and Brendan Lucier.
\newblock Simultaneous auctions are (almost) efficient.
\newblock In {\em Proceedings of the 45th {ACM} Symposium on Theory of
  Computing Conference}, pages 201--210, 2013.

\bibitem{FeldmanGL15}
Michal Feldman, Nick Gravin, and Brendan Lucier.
\newblock Combinatorial auctions via posted prices.
\newblock In {\em Proceedings of the 26th {ACM-SIAM} Symposium on Discrete
  Algorithms}, pages 123--135, 2015.

\bibitem{FeldmanSZ16}
Moran Feldman, Ola Svensson, and Rico Zenklusen.
\newblock Online contention resolution schemes.
\newblock In {\em Proceedings of the 27th {ACM-SIAM} Symposium on Discrete
  Algorithms}, pages 1014--1033, 2016.

\bibitem{HardtRS16}
Moritz Hardt, Ben Recht, and Yoram Singer.
\newblock Train faster, generalize better: Stability of stochastic gradient
  descent.
\newblock In {\em Proceedings of the 33rd International Conference on Machine
  Learning}, pages 1225--1234, 2016.

\bibitem{HartNisan2017}
Sergiu Hart and Noam Nisan.
\newblock Approximate revenue maximization with multiple items.
\newblock {\em Journal of Economic Theory}, 172:313 -- 347, 2017.

\bibitem{KleinbergW12}
Robert Kleinberg and S.~Matthew Weinberg.
\newblock Matroid prophet inequalities.
\newblock In {\em Proceedings of the 44th {ACM} Symposium on Theory of
  Computing Conference}, pages 123--136, 2012.

\bibitem{KrengelS77}
Ulrich Krengel and Louis Sucheston.
\newblock Semiamarts and finite values.
\newblock {\em Bulletin of the American Mathematical Society}, 83:745--747,
  1977.

\bibitem{KrengelS78}
Ulrich Krengel and Louis Sucheston.
\newblock On semiamarts, amarts, and processes with finite value.
\newblock {\em Advances in Probability and Related Topics}, 4:197--266, 1978.

\bibitem{LiYao13}
Xinye Li and Andrew Chi-Chih Yao.
\newblock On revenue maximization for selling multiple independently
  distributed items.
\newblock {\em Proceedings of the National Academy of Sciences of the United
  States of America}, 110(28):11232--11237, 2013.

\bibitem{Roughgarden15}
Tim Roughgarden.
\newblock Intrinsic robustness of the price of anarchy.
\newblock {\em Journal of the {ACM}}, 62(5):32, 2015.

\bibitem{Rubinstein16}
Aviad Rubinstein.
\newblock Beyond matroids: Secretary problem and prophet inequality with
  general constraints.
\newblock In {\em Proceedings of the 48th {ACM} Symposium on Theory of
  Computing}, 2016.
\newblock 324--332.

\bibitem{RubinsteinS17}
Aviad Rubinstein and Sahil Singla.
\newblock Combinatorial prophet inequalities.
\newblock In {\em Proceedings of the 28th {ACM-SIAM} Symposium on Discrete
  Algorithms}, pages 1671--1687, 2017.

\bibitem{rubinstein2018simple}
Aviad Rubinstein and S.~Matthew Weinberg.
\newblock Simple mechanisms for a subadditive buyer and applications to revenue
  monotonicity.
\newblock {\em ACM Transactions on Economics and Computation}, 6(3-4):1--25,
  2018.

\bibitem{SamuelC84}
Esther Samuel-Cahn.
\newblock Comparison of threshold stop rules and maximum for independent
  nonnegative random variables.
\newblock {\em Annals of Probability}, 12:1213--1216, 1984.

\bibitem{SyrgkanisT13}
Vasilis Syrgkanis and {\'{E}}va Tardos.
\newblock Composable and efficient mechanisms.
\newblock In {\em Proceedings of the 45th {ACM} Symposium on Theory of
  Computing}, pages 211--220, 2013.

\bibitem{Vazirani2001}
Vijay~V. Vazirani.
\newblock {\em Approximation Algorithms}.
\newblock Springer-Verlag New York, Inc., New York, NY, USA, 2001.

\bibitem{yao2014n}
Andrew Chi-Chih Yao.
\newblock An n-to-1 bidder reduction for multi-item auctions and its
  applications.
\newblock In {\em Proceedings of the 26th ACM-SIAM symposium on Discrete
  algorithms}, pages 92--109. SIAM, 2015.

\bibitem{Zhang20}
Hanrui Zhang.
\newblock Improved prophet inequalities for combinatorial welfare maximization
  with (approximately) subadditive agents.
\newblock {\em Working paper}, 2020.
\newblock Available from:
  \url{https://users.cs.duke.edu/~hrzhang/papers/subadd_prophet.pdf}.

\end{thebibliography}
